\documentclass[hidelinks,onefignum,onetabnum]{siamart250106}

\usepackage{esint} 
\usepackage{mathtools} 
\usepackage{amsmath,amsfonts,amssymb}
\usepackage{stmaryrd} 
\usepackage{tikz}
\usepackage{tikz-3dplot}
\usepackage{subcaption}
\usetikzlibrary{angles,quotes,calc}

\newtheorem{identity}[theorem]{Identity}

\newtheorem{theoremABC}{Theorem}
\renewcommand{\thetheoremABC}{\Alph{theoremABC}} 

\ifpdf
\hypersetup{
    colorlinks=true,
    linkcolor=blue,
    filecolor=magenta,      
    urlcolor=cyan,
    pdftitle={Fourier-based Inversion of Partial X-ray Transforms in \texorpdfstring{$n$}{n} Dimensions},
    pdfauthor={Murdock G. Grewar},
    pdfpagemode=FullScreen,
}
\fi

\title{Fourier-based Inversion of Partial X-ray Transforms in $n$ Dimensions}

\author{%
Murdock~G.~Grewar\thanks{Department of Materials Physics, Research School of Physics, Australian National University, Canberra, ACT 2600, Australia.}%
}

\headers{Fourier-based Inversion of Partial X-ray Transforms}{Murdock G. Grewar}

\begin{document}

\setlength{\tabcolsep}{4pt}

\maketitle

\begin{abstract}
We present two theorems describing analytic left-inverses of partial X-ray transforms. The first theorem concerns X-ray data collected with an arbitrary distribution of parallel projections; it contains a convolution-backprojection formula and a backprojection-convolution formula for recovering the transformed volume, provided the data is sufficient. The second theorem concerns X-ray data collected with a cone-beam; it contains a backprojection-convolution formula for recovering the transformed volume, provided the data is amenable to this method (for example: $(n-1)$-dimensional source loci that `surround' the reconstruction support; detectors of finite size are supported). These theorems are the outcome of a modestly general and rigorous investigation undertaken into the existence of backprojection-convolution methods in $n-$dimensional space. Necessary and sufficient conditions on the experiment geometry are established for the existence of such methods, as are the particular error metrics minimised by backprojection-convolution methods and the uniqueness of those minimum-error solutions. A major practical outcome of this work is the production of the first known exact inversion methods for cone-beam geometries where the X-ray source point loci are multidimensional, such as (in 3D) a cylinder or a sphere of X-ray source positions. A separate article describes a practical computer implementation for the case of a cylinder in 3D.
\end{abstract}

\begin{keywords}
computed tomography, x-ray transform, filtered backprojection, backprojection-filtration, exact reconstruction
\end{keywords}

\begin{MSCcodes}
45Q05, 
92C55, 
44A05 
\end{MSCcodes}

\section{Introduction}\label{sec:introduction}

The X-ray transform maps a function $\mu : \mathbb R^n \rightarrow \mathbb R$ into the collection of all the line integrals of that function \cite{natterer2001mathematics}. The transform is so named for its use in computerised tomography (CT), wherein it maps a 3-dimensional attenuating body to a space of theoretical measurements which could be obtained by passing X-rays through that body. In that case, $\mu$ is termed the `attenuation coefficient' field, and the X-ray transform is a theoretical model for X-ray transmission experiments. A left-inverse to the X-ray transform recovers from the measurement data the volume from whence it came---this kind of reconstruction is the basic goal of transmission CT, and provides the motivation for the theoretical study of inverses to the X-ray transform.

In transmission CT experiments, measurements are not taken along every possible line through a volume--- the X-ray transform is restricted to the lines of measurement included in the experiment, hence our nomenclature `partial X-ray transform'. There are two prototypical varieties of partial X-ray transforms relevant to experiments:
\begin{itemize}
    \item what we will call \emph{parallel-beam} X-ray transforms, wherein a line is included if and only if every parallel line is also included.
    \item what we will call \emph{cone-beam} X-ray transforms, wherein a line is included if and only if it intersects $X \subset \mathbb R^n$, where $X$ is predetermined.
\end{itemize}
The former kind corresponds to experiments conducted at high-energy facilities (e.g. synchrotrons), where a parallel beam of X-rays is used to scan the object at a variety of angles. The latter kind corresponds to experiments conducted in the laboratory, and in medical and industrial settings, and the set $X$ is the locus of points at which the X-ray source is placed and measurements through the object collected. A typical choice of $X$ is a helix wrapping around the object. We are simplifying the latter case somewhat, as in reality the cone-beam transforms are further restricted by the finite size of the X-ray detector (we have accounted for this in our inversion method).

The study of inversion theory of the X-ray transform, or partial versions of it, is perhaps best known through the work of \emph{A. M. Cormack} \cite{cormack1963, cormack1964, cormack1982}, dating 1963 onwards which, together with the work of \emph{G. N. Hounsfield} on imaging the human brain \cite{hounsfield1980computed}, led to a shared 1979 \emph{Nobel Prize in Physiology or Medicine ``for the development of computer assisted tomography''}.
Initially, analytic inversion theory was confined to reconstruction from parallel-beam data. (Or, of course, data from a cone-beam or other sort of geometry may be rearranged into a parallel-beam format, at large computational expense.) The parallel-beam inversion theory is relatively simple due to one \emph{Fourier slice theorem}.\footnote{
   The earliest application of the Fourier Slice Theorem to signal analysis is not known by us. 
   An early reference is \cite{bracewell1956strip}.
} Besides this, it seems numeric methods were used to tackle the inverse problem more-or-less as a large but generic linear algebra problem \cite{gordon1970algebraic} (ART; 1970), which was later improved to minimise a certain error function of the reconstruction at an enormous expense of computation time\footnote{
The relative computational expense of SIRT is related to the size of the reconstruction, and it may have been very adequate at the time.} \cite{gilbert1972iterative} (SIRT; 1972); see also \cite{goitein1972three} (1972) for a more generic error minimisation algorithm, and \cite{scales1987tomographic} for the `conjugate gradient method' (CGM; 1987) applied in tomography. Much later, a method was published to speed up ART by exploiting parallel-beam geometry \cite{andersen1984simultaneous} (SART; 1984). To combat measurement noise, an alternative method of reconstruction was proposed which again came at great computational expense \cite{lange1984reconstruction} (EMTR; 1984). 

Notably, there also came a very popular reconstruction method for cone-beam scans \cite{feldkamp1984practical} (FDK; 1984) in which the scanning locus $X$ is a circle. In this author's opinion, the outstanding contribution of \cite{feldkamp1984practical} was to rephrase an impractical reconstruction formula derived from the Radon transform \cite{radon1917_english} into a convolution-backprojection formula, which is fast to compute on modern hardware due to the Fast Fourier Transform. Though the FDK method is only able to correctly reconstruct that portion of the object residing in the plane containing the circle (cf. sufficiency conditions proven in \cite{smith1985}), the method remains widely known and is still sometimes used to produce (inaccurate) 3D reconstructions because of its simplicity and because its convolution-backprojection structure grants exceptional computational performance. 
Another highly-influential result is that of \emph{Colsher} in 1980 \cite{colsher1980fully} for 3D reconstruction in Positron Emission Tomography (PET) \cite{defrise2005image}. That result provides an inversion method for a particular variety of parallel-beam X-ray transforms, and is appealing for its backprojection-convolution structure.

It wasn't until 1987--1991 when a theoretical development was made by \emph{Grangeat} \cite{grangeat1987analyse,grangeat1991mathematical} linking cone-beam projection data to a derivative of the Radon transform \cite{radon1917_english} that fully 3D cone-beam reconstruction became common, and has since been referred to as cone-beam computed tomography (CBCT). In mathematical terms: \emph{Grangeat} established a singular-value decomposition for the projection operator which had left- and right-singular vectors that related to the Radon transforms of the measurements and volume respectively. 
An exact method for reconstruction from cone-beam data was made available in 1995 \cite{tam1995exact} based on this theoretical development, after which many more surfaced, e.g. \cite{danielsson97:proceedings_3D_1997,defrise2000solution,katsevich2002theoretically,katsevich2003general,ye2005general}.
Based as they are on the work of \emph{Grangeat}, these methods are limited to source loci $X$ that are one-dimensional curves through 3D space, because they rely on differentiation of measurement data along the parameter of that curve. 
Once again, the method that has best stood the test of time is structured similarly to a convolution-backprojection algorithm \cite{katsevich2002theoretically,katsevich2003general} (typically it is referred to as a `filtered backprojection' method, because `filtered' can refer to more general operations than a convolution).
Indeed, it is hard to imagine any type of inversion algorithm that could exceed the speed of convolution-backprojection type methods at large scale. 

\begin{figure}[t]
\begin{center}
    \begin{subfigure}{0.23\linewidth}
    \centering
    \includegraphics[width=0.8\linewidth]{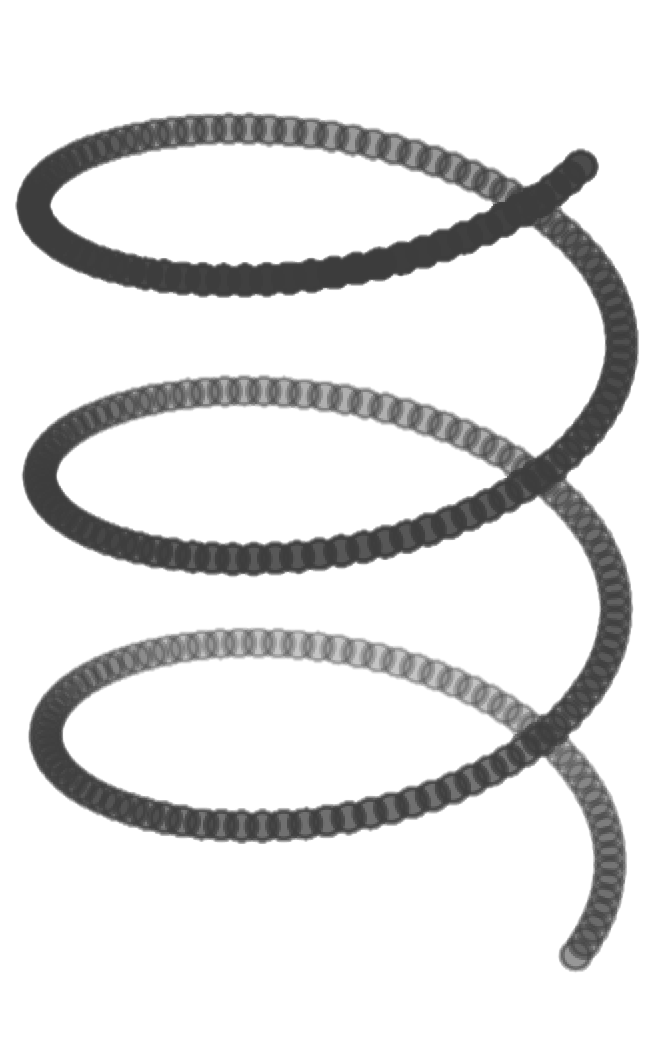}
    \caption{a helix}
    \label{fig:trajectories_helix}
    \end{subfigure}%
    \begin{subfigure}{0.23\linewidth}
    \centering
    \includegraphics[width=0.8\linewidth]{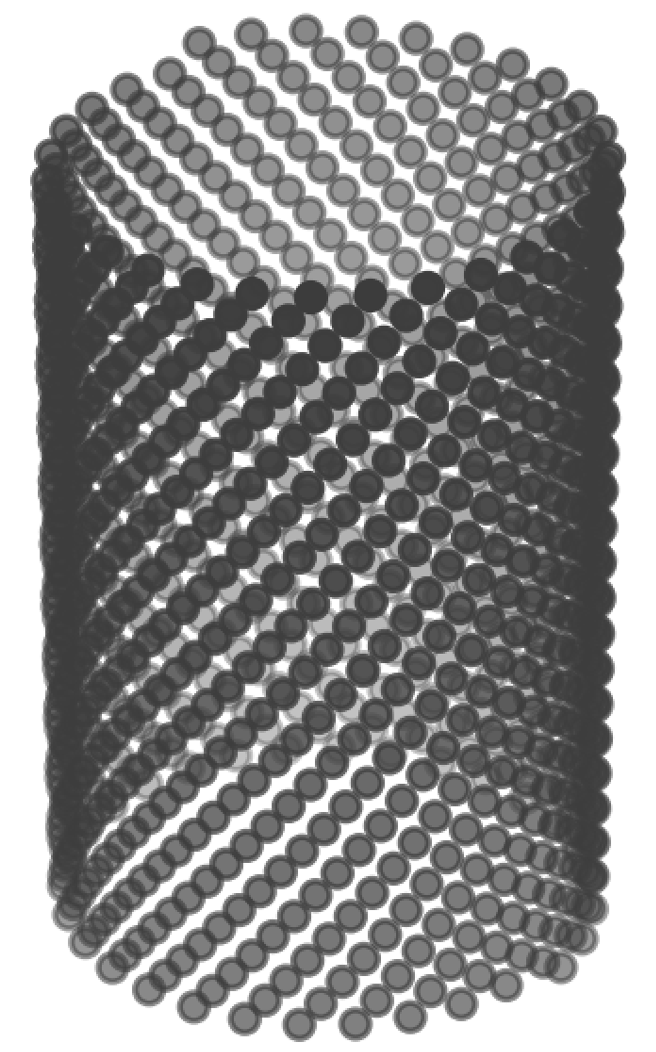}
    \caption{a cylinder (\cite{kingston2018space})}
    \label{fig:trajectories_sft}
    \end{subfigure}%
    \begin{subfigure}{0.23\linewidth}
    \centering
    \includegraphics[width=0.8\linewidth]{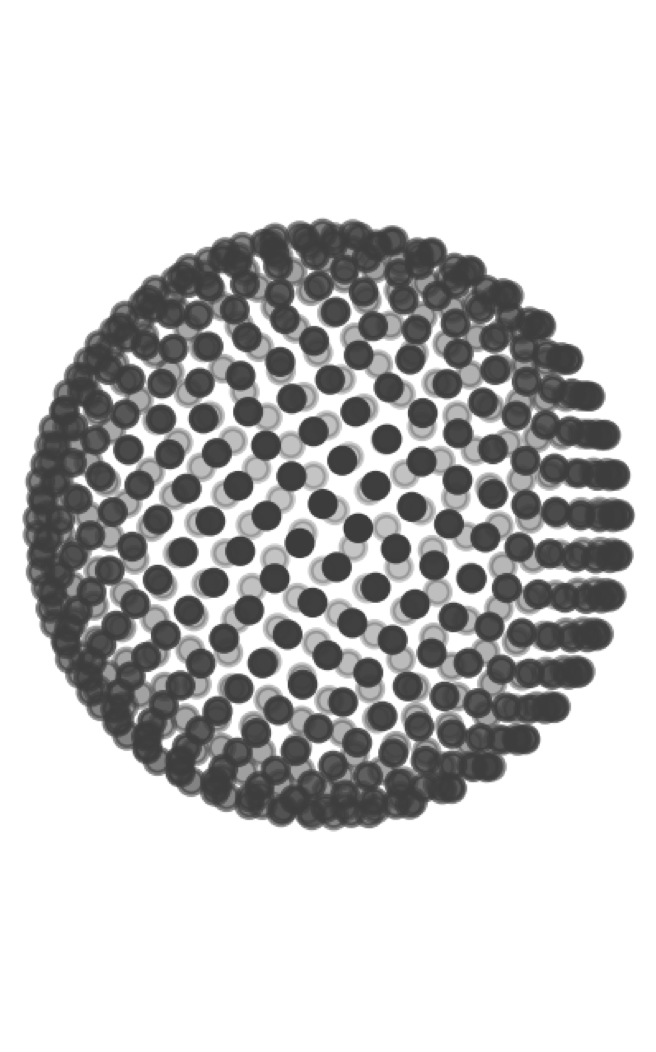}
    \caption{a sphere}
    \label{fig:trajectories_lds_sphere}
    \end{subfigure}
    \caption{Examples of acquisition trajectories in 3D CBCT. The discrete points are the locations at which the X-ray source is placed as transmission measurements are made. The latter two trajectories are multidimensional and the associated measurement data are amenable to direct reconstruction using theorem~\ref{th:cone_gbpf}. An example implementation of a reconstruction algorithm operating on the cylindrical locus is given in \cite{grewar2024preprint}.}
    \label{fig:trajectories}
\end{center}
\end{figure}

If one were to choose the ideal source locus $X$ based on purely geometric considerations, i.e. without being constrained by the computational difficulty of the inversion problem or the logistics of physically acquiring the data, then one would choose $X$ to take on a shape that spreads itself more-or-less evenly around the object. This ensures greater statistical independence between measurements, and therefore greater data/X-ray-dose efficiency. In 3D CBCT, this implies that $X$ should be a $2-$dimensional surface residing in 3D space. This has been explored previously in the literature, e.g. with $X$ shaped like the surface of a cylinder (see fig.~\ref{fig:trajectories_sft}) in \cite{noo1997} (1997) and more recently in \cite{kingston2018space} (2018), or with $X$ shaped like the surface of a sphere (see fig.~\ref{fig:trajectories_lds_sphere}) in \cite{bauer2021} (2021). However, the drawback of such geometries is that---until now---an analytic method of inversion has not been known. Therefore, no fast method of reconstruction has been available, and iterative methods have generally been used, such as \cite{myers2016rapidly} for the cylinder. 

Indeed, a notable absence from the literature is any analytic method of reconstruction from cone-beam data collected with a source locus $X$ that is more general than a one-dimensional curve, e.g. the aforementioned cylinder or sphere, despite the data efficiency of such trajectories. 
Similarly for the parallel-beam case, we are unaware of a published method that describes how to directly reconstruct from parallel-beam X-ray transforms with arbitrary distributions of beam directions, e.g. an even sampling of directions around the sphere. (We are only aware of the particular cases of the circular scan and the `spheroannular' scan in \cite{colsher1980fully}.) 

It is clear from historical precedent and the existence of the Fast Fourier Transform that convolution-backprojection and backprojection-convolution style methods are performant. 
This article undertakes a modestly general and rigorous examination of Fourier-based inversion of parallel and cone-beam partial X-ray transforms in $n$ dimensions, for all integers $n \geq 2$. We identify the precise conditions on the experiment geometry which enable backprojection-convolution methods to be constructed (and conversely, when such methods cannot be constructed). In mathematical terms: we explore the requirements on the experiment geometry that admit Fourier singular-value decompositions of partial X-ray transforms, i.e. singular-value decompositions wherein the left- and right-singular vectors are Fourier plane waves in the measurement and volume spaces respectively.
We also specify a class of error metrics that can be minimised by these methods,\footnote{We have not explored the possible error metrics exhaustively. As a particular example: there may be additional available error metrics that can be minimised with backprojection/convolution methods in the case of parallel-beam circular scans for $n\geq2$.} and prove the uniqueness of these minimum-error solutions. For the parallel-beam case, we also do the same for convolution-backprojection methods. Our examination culminates in two theorems: one for the parallel-beam case, and one for the cone-beam.
The theorems directly yield three varieties of new inversion methods:
\begin{itemize}
    \item convolution-backprojection methods for arbitrary parallel-beam X-ray data.
    \item backprojection-convolution methods for arbitrary parallel-beam X-ray data.
    \item backprojection-convolution methods for certain kinds of cone-beam datasets (particularly, for source loci $X$ which are multidimensional and surround the reconstruction support).
\end{itemize}
One major practical outcome of this investigation is the category of backprojection-convolution methods for cone-beam data. For the avoidance of doubt: these methods are compatible with real experiments, where the X-ray detector surface is a rectangle of finite size. 

In research separate to this article, we have produced a performant direct inversion algorithm for the geometry where $X$ is a cylinder, based on the theory developed herein. That involves the introduction of several `discretisation regularisations' and extends well beyond the scope of the present article. As expected, the method is efficient and accurate. That work is available in \cite{grewar2024preprint}. 


\subsection{The problem statement}\label{sec:intro_error}

The basic theoretical model of transmission CT is determined by the Bouguer-Beer-Lambert law \cite{mayerhofer2020bouguer}. In this model, the vector of line attenuations $m$---otherwise known as the \emph{linearised measurements}---may be obtained experimentally. These measurements are related to the volume by $m = A \mu$, where $A$ is a linear integral operator and a partial form of the total X-ray transform.

The central problem of transmission CT is to reconstruct $\mu$ from our knowledge of $A$ (the experiment geometry) and $m$ (the transformed measurement data). In other terms, the problem is to produce a left-inverse to the operator $A$---or at least a computer algorithm that implements such an operator---which may then be applied to $m$ to yield $\mu$. 
The inverse problem is soluble only when the experiment geometry, encoded by $A$, includes a sufficient collection of measurement lines for the assumed support of $\mu$ \cite{smith1985}. 
When it is soluble, there still remains ambiguity in the proper estimate for $\mu$ because the measurements $m$ are collected experimentally and therefore include measurement errors \cite{nuyts2013modelling} (i.e. $m \not \in \mathrm{range}(A)$ due to noise in $m$).
For this reason, it is proper that we investigate how the reconstruction may be chosen to minimise chosen error metrics. We find that for certain error metrics and imaging geometries, the reconstruction can be produced via backprojection-convolution methods (and convolution-backprojection methods in the parallel-beam case).

\vfill
\pagebreak
\subsection{Overview of the new inversion theory}

The theoretical framework developed in this article culminates in two major results. These results are given formally as the two theorems in \S\ref{sec:statement_of_results}. In \S\ref{sec:examples}, we demonstrate how the theorems may be applied to produce exact inversion methods, e.g. for the cylindrical and spherical scanning trajectories in cone-beam tomography. Finally, the theorems are proven in \S\ref{sec:proofs}.

Let us now provide a brief overview of how the theorems are derived. 
First, suppose we define an inner product $\langle \cdot , \cdot \rangle$ on the space of possible measurement vectors, i.e. the space of which $m$ is a member. We define our left-inverse $A^{-L}$ as that operator which minimises an error metric dependent on $\langle \cdot , \cdot \rangle$ viz.
\begin{equation}
    A^{-L} m = \mathrm{argmin}_x ||m - Ax|| = \mathrm{argmin}_\mu \langle m - Ax , m - Ax \rangle \, .
\end{equation}
The solution, which is proven in lemma~\ref{lem:min_norm}, is $A^{-L} = (A^\dag A)^{-1} A^\dag$. This operator is known more widely in generic linear algebra terms as the Moore-Penrose left pseudoinverse of $A$ (e.g. see \cite{ben2006generalized}). It depends on the choice of inner product, because the inner product defines the adjoint operator $A^\dag$. This adjoint is often called the \emph{backprojection operator}. 

First, we consider the case of \emph{parallel-beam} partial X-ray transforms.
In $3$ dimensions, we refer to the measurement data parallel to a given direction as a `2D projection set'. It is already known (e.g. see \cite{kak2001principles}) from the Fourier Slice Theorem that, assuming a `reasonable' $L^2$-like inner product $\langle \cdot , \cdot \rangle$, the backprojection of a 2D projection set results in a volume whose Fourier transform is the essentially the same as that of the 2D projection set, except embedded as a 2D plane within the 3D Fourier space of the volume. A variant of this statement is proven in lemma~\ref{lem:fst}, for a large class of inner products. By summing over all the projection directions, it follows that the Fourier transform of the full backprojection $A^\dag m$ is a summation of these planes; this is described precisely in corollary~\ref{cor:ATA_decomp}, for the aforementioned `large class of inner products'. This corollary is then specialised to those particular `reasonable' inner products, resulting in the more specialised but useful result in lemma~\ref{lem:ATA_decomp_simplified}. 
Notably, lemma~\ref{lem:ATA_decomp_simplified} provides an identity for $A^\dag A$ as a volume convolution with a kernel that is determined by an $n$-dimensional generalisation of the Funk transform \cite{funk1911flachen}. 
With lemma~\ref{lem:ATA_decomp_simplified} established, the inversion formula $A^{-L} = (A^\dag A)^{-1} A^\dag$ is understood as a backprojection-convolution method. To convert it to a convolution-backprojection method, identity~\ref{id:vol_proj_filt} is established, which facilitates reordering of the convolution so that it occurs in projection/measurement space, prior to backprojection.
In this manner, a general convolution-backprojection formula is derived for inverting an arbitrary \emph{parallel-beam} partial X-ray transform, provided it is in fact invertible. The specific choice of inner product affects both the convolution kernel and the error function minimised by the resulting reconstruction.
The specifics are codified in theorem~\ref{th:parallel_gbpf}.

Finally, we consider the case of \emph{cone-beam} partial X-ray transforms. We deal with only a restricted class of such transforms. In particular: we assume that the measurement data can be `rebinned' into parallel-beam data, because in this case it becomes possible to construct backprojection-convolution methods of reconstruction. The rebinning is not an explicit operation in the method, but is accomplished implicitly as part of the backprojection procedure by incorporating a weighting term that doesn't incur significant complexity or computational expense. For this reason, backprojection must precede convolution, so we revert to the backprojection-convolution method. Also, our theorem for the cone-beam case `bakes in' an important discretisation regularisation that is necessary for numerical stability when the method is implemented in a computer algorithm. The inversion formula for these cone-beam experiments is codified in theorem~\ref{th:cone_gbpf}.



\pagebreak

\subsection{Notation}\label{sec:notation}
Throughout this article, the following notational conventions are followed.

\begin{itemize}
\item The Euclidean inner product between vectors $r_1, r_2 \in \mathbb R^n$ is denoted $r_1 \cdot r_2$.

\item The sphere $S^{n-1}$ is considered to be the subset of $\mathbb R^n$ consisting of vectors with unit length, under the usual Euclidean norm. We use the vectors in $S^{n-1}$ to specify the directions of projection lines. 

\item \textbf{Anonymous functions.} Consider the function that squares a number and adds $1$. Conventionally, this function would be defined and denoted viz.
    \begin{equation*}
        \mathrm{sqp1} : \mathbb R \rightarrow \mathbb R^+, \qquad \mathrm{sqp1}: x \mapsto x^2 + 1\, .
    \end{equation*}
    This article benefits from a more powerful and succinct notation. We write the function like so:
    \begin{equation*}
        \mathrm{sqp1} = \llbracket \mathbb R \rightarrow \mathbb R^+ : x \mapsto x^2+1\rrbracket \quad \text{or} \quad \mathrm{sqp1} = \llbracket_{\mathbb R}^{\mathbb R^+} x \mapsto x^2+1 \rrbracket \, ,
    \end{equation*}
    or simply $\mathrm{sqp1} = \llbracket x \mapsto x^2+1 \rrbracket$ where the domain and codomain are assumed to be clear from context.
\end{itemize}

\section{Statement of results}\label{sec:statement_of_results}
The main results of this article are theorems~\ref{th:parallel_gbpf},~\ref{th:cone_gbpf}.
First we list definitions in \S\ref{sec:definitions_01} that are required to state the theorems concisely. The theorems are presented in \S\ref{sec:theorems_01}.
Examples of how to use these theorems to construct inversion algorithms are given in \S\ref{sec:examples}. 
All proofs are contained in \S\ref{sec:proofs}.

\subsection{Definitions}\label{sec:definitions_01}
We label the space of attenuating volumes $\mathcal A$, and the space of measurement sets $\mathcal M$. We define the X-ray transform, $A : \mathcal A \rightarrow \mathcal M$, and an $n$-dimensional generalisation of the Funk transform \cite{funk1911flachen}, which we call $\mathcal I$.

The set of directed lines is labelled $M$. A directed line is specified by the tuple $(\hat \theta, r)$ which encodes the direction $\hat \theta \in S^{n-1}$ and a point $r \in \mathbb R^n$ included in the line.
\begin{definition}[$M$: measurement line space]
    $M$ is the space of directed lines through $\mathbb R^n$. It is the quotient of $S^{n-1} \times \mathbb R^n$ by the equivalence relation
    \begin{equation}
        (\hat \theta_1, r_1) \equiv (\hat \theta_2, r_2) \quad \mathrm{iff} \quad \left( \hat \theta_1 = \hat \theta_2 \quad \mathrm{and} \quad r_1 - \hat \theta_1 (\hat \theta_1 \cdot r_1) = r_1 - \hat \theta_2 (\hat \theta_2 \cdot r_2) \right) \, .
    \end{equation}
\end{definition}
We will sometimes need to refer to the subset of lines that point in a common direction.
\begin{definition}[$M_{\hat \theta}$: single-direction measurement line space]
    $M_{\hat \theta}$ with $\hat \theta \in S^{n-1}$ is defined as the space of directed lines through $\mathbb R^n$ that are parallel to $\hat \theta$. It is constructed as the quotient of $\mathbb R^n$ by the equivalence relation
    \begin{equation}
        r_1 \equiv r_2 \quad \mathrm{iff} \quad  
        \quad r_1 - \hat \theta(\hat \theta \cdot r_1) = r_2 - \hat \theta(\hat \theta \cdot r_2) \, .
    \end{equation}
\end{definition}
We may now define $\mathcal A$ and $\mathcal M$.
\begin{definition}[$\mathcal A$: attenuation field space]\label{def:atten_field_space_unweighted}
    $\mathcal A$ is the complex Hilbert space consisting of functions $\mathbb R^n \rightarrow \mathbb C$ and equipped with the $L^2$ inner product. 
\end{definition}
\begin{definition}[$\mathcal M_{\hat \theta}$]
    $\mathcal M_{\hat \theta}$ is defined as the complex Hilbert space consisting of functions $M_{\hat \theta} \rightarrow \mathbb C$ equipped with the $L^2$ inner product.
\end{definition}
\begin{definition}[$\mathcal M$: measurement field space]
    $\mathcal M$ is defined as the complex Hilbert space consisting of functions $M \rightarrow \mathbb C$ and equipped with the $L^2$ inner product.
\end{definition}

The eponymous \emph{(total) X-ray transform} is defined in our formalism as follows:
\begin{definition}[$A$: X-ray transform]\label{def:xray_transform}
    The (total) X-ray transform, denoted $A$, is defined as the following linear integral operator:
    \begin{align}
    A = \Biggl\llbracket_{\mathcal A}^{\mathcal M} \mu \mapsto \Bigl\llbracket_M^{\mathbb C} (\hat \theta, r) \mapsto \int_{-\infty}^{\infty} \mathrm d s \mu(r + s \hat \theta) \Bigr\rrbracket \Biggr\rrbracket \, .
    \end{align}
\end{definition}

We require a few ancillary definitions so that we may concisely state the reconstruction formulae given by theorems~\ref{th:parallel_gbpf},~\ref{th:cone_gbpf}. 
Frequently, we will write integrations that take place over $(n-1)$-dimensional subspaces of $\mathbb R^n$ that are orthogonal to a given direction $\hat \theta \in S^{n-1}$. This motivates the following definition.
\begin{definition}[$\hat \theta^\perp$: subspace orthogonal to $\hat \theta$]
    $\hat \theta^\perp$, where $\hat \theta \in S^{n-1}$, is defined as
    \begin{equation}
    \hat \theta^\perp = \left\lbrace r \in \mathbb R^n : r \cdot \hat \theta = 0 \right\rbrace \, .
    \end{equation}
\end{definition}
We next define an $n$-dimensional generalisation of the Funk transform \cite{funk1911flachen}. This transform takes a distribution on $S^{n-1}$ and produces a new one. 
\begin{definition}[$\mathcal I$: generalised Funk transform]
    $\mathcal I[d]$, where $d$ is a real-valued distribution over $S^{n-1}$, is defined as
    \begin{equation}
        \mathcal I[d] = \Biggl\llbracket_{S^{n-1}}^{\mathbb R} 
            \hat \theta \mapsto \oint_{\mathrlap{\{\hat \theta' \in S^{n-1} : \hat \theta' \cdot \hat \theta = 0\}}} \; \mathrm d^{n-2} \hat \theta' \;\;\;\, d(\hat \theta')
        \Biggr\rrbracket\, .
    \end{equation}
    The integration takes place over the vectors of unit length that are orthogonal to $\hat \theta$. This is an integration over a $(n-2)$-dimensional sphere, with the usual integration measure inherited from $\mathbb R^n$. In the case of $n = 2$, the integration is understood as a summation over the two vectors of unit length that are orthogonal to $\hat \theta$.
\end{definition}
As will be expressed in the theorems in \S\ref{sec:theorems_01}, data sufficiency mandates that the distribution $d$ be chosen such that $\mathcal I[d](\hat \theta) \neq 0$ for all $\hat \theta \in S^{n-1}$. The choice of $d$ is also constrained to be compatible with the experiment geometry, but is otherwise free to be chosen to influence the error metric that will be minimised by the reconstruction. To be compatible with the experiment geometry, the support of $d$ must be contained within the set of line directions for which measurement data is actually available. In this manner, the theorems apply not only to the total X-ray transform, but also partial X-ray transforms.

We also have a series of Fourier transforms which are applied to the volume and measurement data. For formal correctness, we need to define separate transformations to operate on the different domains. 
\begin{definition}[$F_{\mathcal A}$: Volume Fourier Transform]\label{def:vol_fourier_transform}
    We define the operator $F_{\mathcal A} : \mathcal A \rightarrow \mathcal A$, and its inverse, as the unitary $n-$dimensional Fourier transform:
    \begin{align}
        F_{\mathcal A}^{\pm 1} = \Biggl\llbracket_{\mathcal A}^{\mathcal A}
                \mu &\mapsto \Bigl\llbracket_{\mathbb R^n}^{\mathbb C}
            k \mapsto (\sqrt{2\pi})^{-n} \int_{\mathbb R^n} \mathrm d^n r \,
                \mu(r) e^{\mp i r \cdot k}
        \Bigr\rrbracket \Biggr\rrbracket \, .
    \end{align}
\end{definition}
\begin{definition}[$F_{\hat \theta}$: Single-Direction Projection Fourier Transform]\label{def:spd_ft}
    Let $\hat \theta \in S^{n-1}$. 
    We define the operator $F_{\hat \theta} : \mathcal M_{\hat \theta} \rightarrow \mathcal M_{\hat \theta}$, and its inverse, as the unitary $(n-1)$-dimensional Fourier transform:
    \begin{align}
        F_{\hat \theta}^{\pm 1} = \Biggl\llbracket_{\mathcal M_{\hat \theta}}^{\mathcal M_{\hat \theta}} 
        p &\mapsto 
        \Bigl\llbracket_{\hat \theta^\perp}^{\mathbb C} k_\perp \mapsto (\sqrt{2\pi})^{-(n-1)} \int_{\hat \theta^\perp}
        \mathrm d^{n-1} r_\perp p(r_\perp) e^{\mp ir_\perp \cdot k_\perp}
        \Bigr\rrbracket
        \Biggr\rrbracket \, .
    \end{align}
\end{definition}
\begin{definition}[$F_{\mathcal M}$: Projection Fourier Transform]\label{def:full_proj_ft}
    We define the operator $F_{\mathcal M} : \mathcal M \rightarrow \mathcal M $, and its inverse, as a unitary $(n-1)$-dimensional Fourier transform. It operates like $F_{\hat \theta}$, except it does so on all directions $\hat \theta$ at once. Formally, $F_{\mathcal M}$ is defined by:
    \begin{align}
        F_{\mathcal M}^{\pm 1} = \Biggl\llbracket_{\mathcal M}^{\mathcal M}
        m &\mapsto \Bigl\llbracket_{M}^{\mathbb C} (\hat \theta, k_\perp) \mapsto (\sqrt{2\pi})^{-(n-1)} \int_{\hat \theta^\perp}
        \mathrm d^{n-1} r_\perp m(\hat \theta, r_\perp) e^{\mp ir_\perp \cdot k_\perp}
        \Bigr\rrbracket
        \Biggr\rrbracket \, .
    \end{align}
\end{definition}

We also define a notation to express diagonal operators.
\begin{definition}[Diagonal operators]\label{def:diag}
    We use the following notation to define diagonal operators. The domains and codomains are to be understood from context.
    \begin{align}
        \mathrm{diag}_x\left( \underbrace{f(x)}_{\text{an expression involving $x$}} \right) = \Biggl\llbracket 
            g \mapsto \Bigl\llbracket y \mapsto f(y) g(y) \Bigr\rrbracket
        \Biggr\rrbracket \, .
    \end{align}
\end{definition}

\subsection{Theorems}\label{sec:theorems_01}
The two theorems below comprise the main results of this article. Their proofs are contained in \S\ref{sec:proofs}. Examples of how they can be applied to transmission tomography experiments are given in \S\ref{sec:examples}.
\subsubsection{Theorem for parallel-beam inversion}
\begin{theoremABC}[A left-inverse for parallel-beam partial X-ray transforms]\label{th:parallel_gbpf}
    Let $m \in \mathcal M$.
    Let $d$ be a real, nonnegative distribution on $S^{n-1}$.

    If and only if the generalised Funk transform of $d$ is everywhere-nonzero, i.e. 
    \begin{equation}\label{eq:th:gbpf_existence}\tag{\thetheoremABC:sufficiency}
        \forall \hat \theta \in S^{n-1}, \qquad \mathcal I[d](\hat \theta) \neq 0 \, , 
    \end{equation}
    then there is a unique $\mu \in \mathcal A$ satisfying the minimum-error equation
    \begin{equation}\label{eq:th:gbpf_err}\tag{\thetheoremABC:objective}
        \mu = \mathrm{arg~min}_{\mu' \in \mathcal A} \oint_{\mathrlap{S^{n-1}}} \; \mathrm d^{n-1} \hat \theta \int_{\mathrlap{\hat \theta^\perp}} \mathrm d^{n-1} r_\perp \; d(\hat \theta) \left(\left(m - A\mu'\right)(\hat \theta, r_\perp) \right)^2 \, ,
    \end{equation}
    given by $\mu = A^{+_d} m$ where $A^{+_d}$ can be expressed as a backprojection-convolution as follows:
    \begin{equation}\label{eq:th:gbpf_inv_formula}\tag{\thetheoremABC:BPC}
        A^{+_d} = 
        \underbrace{
        F_{\mathcal A}^{-1} \mathrm{diag}_k\left(\frac{|k|}{2 \pi \mathcal I[d](\hat k)}
            \right) F_{\mathcal A}
        }_{\text{volume deconvolution}} \;
            \underbrace{\Biggl\llbracket_{\mathcal M}^{\mathcal A} p \mapsto \Bigl\llbracket_{\mathbb R^n}^{\mathbb C}
                r \mapsto \oint_{\mathrlap{S^{n-1}}} \; \mathrm d^{n-1} \hat \theta \; d(\hat \theta) p(\hat \theta, r) \Bigr\rrbracket
            \Biggr\rrbracket}_{\text{weighted backprojection}} \, ,
    \end{equation}
    where $\hat k = k/|k|$.
    The operator $A^{+_d}$ may also be expressed as a convolution-backprojection, as follows:
    \begin{equation}\label{eq:th:fbp_inv_formula}\tag{\thetheoremABC:CBP}
        A^{+_d} = \underbrace{\Biggl\llbracket_{\mathcal M}^{\mathcal A} p \mapsto\Bigl\llbracket_{\mathbb R^n}^{\mathbb C}
                r \mapsto \oint_{\mathrlap{S^{n-1}}} \; \mathrm d^{n-1} \hat \theta \; d(\hat \theta) p(\hat \theta, r)
            \Bigr\rrbracket\Biggr\rrbracket}_{\text{weighted backprojection}}
        \underbrace{F_{\mathcal M}^{-1} \mathrm{diag}_k\left(\frac{|k|}{2 \pi \mathcal I[d](\hat k)}
            \right) F_{\mathcal M}}_{\text{projection deconvolution}} \, .
    \end{equation}
\end{theoremABC}
Example applications of theorem~\ref{th:parallel_gbpf} are given in \S\ref{sec:examples_parallel}. It is seen from those examples that the reconstruction formulae \eqref{eq:th:gbpf_inv_formula} and \eqref{eq:th:fbp_inv_formula} reduce to previously-known formulae in the appropriate circumstances.

\subsubsection{Theorem for cone-beam inversion}

The cone-beam case is mathematically not very different from the parallel-beam case. The distinction is the manner in which the X-ray transform is restricted, and the way the data is arranged in practice. Theorem~\ref{th:parallel_gbpf} can be applied to cone-beam data if there is sufficient measurement data collected, $m(\hat \theta, r_\perp)$, over a variety of projection lines that contain the equivalent data from a parallel-beam transform. There is nothing fundamentally different about the cone-beam variant of the theorem, in this sense.

In order to appreciate the difference of the cone-beam variant of the theorem, practical concerns must be considered. The measurement data is no longer naturally arranged in the format $m(\hat \theta, r_\perp)$ because the measurement lines from an X-ray source point diverge; the data is collected on a 2D detector which measures these divergent rays. So, the $2$-dimensional convolution over parallel measurement lines demanded by the convolution-backprojection formula \eqref{eq:th:fbp_inv_formula} is actually a rather complicated transformation of the real underlying data format. Thus, the practical appeal of convolution-backprojection is largely lost in the cone-beam case. Fortunately, the backprojection-convolution formula \eqref{eq:th:gbpf_inv_formula} can be salvaged.

To facilitate the practical application of \eqref{eq:th:gbpf_inv_formula} to data collected with a cone-beam source, we reparameterise the measurement lines from $m(\hat \theta, r)$ to $m(\hat \theta, x)$ which refers to the directed measurement line pointing in the direction of $\hat \theta \in S^{n-1}$ and containing $x \in \mathbb R^n$, where $x$ varies over the cone-beam source points. 
The expression for the weighted backprojection in theorem~\ref{th:parallel_gbpf}, applied to the measurement field $m$, is
\begin{equation}
            \Biggl\llbracket_{\mathbb R^n}^{\mathbb C} r \mapsto \oint_{\mathrlap{S^{n-1}}} \; \mathrm d^{n-1} \hat \theta \; d(\hat \theta) m(\hat \theta, r) \Biggr\rrbracket \, .
\end{equation}
With parallel-beam data, it is natural to integrate over $\hat \theta$ computationally, because the 2D measurement sets each correspond to a fixed projection direction $\hat \theta_i$, and so an integration over $\hat \theta$ amounts to a summation over the 2D measurement sets. In the cone-beam case however, each 2D measurement set corresponds to a source point $x$. 
For this reason, we seek to change the integration coordinates from $\hat \theta \in S^{n-1}$ to the source points $x \in \mathbb R^n$. The outcome of this process is essentially the content of theorem~\ref{th:cone_gbpf}.

\begin{theoremABC}[A left-inverse for cone-beam partial X-ray transforms]\label{th:cone_gbpf}
    Let $m_{\pm} \in \mathcal M$ and let $m_{\pm}$ satisfy the symmetry: $m_{\pm}(\hat \theta, r) = m_{\pm}(-\hat \theta, r)$. 
    Let $d_{\pm}$ be a nonnegative real distribution on $S^{n-1}$ and let $d_{\pm}$ satisfy the symmetry: $d_{\pm}(\hat \theta) = d_{\pm}(-\hat \theta)$.
    Let $\chi$ be a nonnegative real distribution on $\mathbb R^n$ (to be regarded as a density function of source points).
    Assume that the `amenability condition' holds:
    \begin{equation}\label{eq:th:cone_gbpf_amenability}\tag{\thetheoremABC:amenability}
        \forall r \in \mathbb R^n, \hat \theta \in S^{n-1} \quad \left(
            \int_{-\infty}^\infty \mathrm d s \; |s|^{n-1} \chi(r + s \hat \theta) = 0 
            \quad \Rightarrow \quad 
            d_{\pm}(\hat \theta) m_{\pm}(\hat \theta, r) = 0 
        \right) \, .
    \end{equation}
    Then if and only if the generalised Funk transform of $d_{\pm}$ is everywhere-nonzero, i.e.
    \begin{equation}\label{eq:th:cone_gbpf_existence}\tag{\thetheoremABC:sufficiency}
        \forall \hat \theta \in S^{n-1}, \qquad 0 \neq \mathcal I[d_{\pm}](\hat \theta) \, ,
    \end{equation}
    then there is a unique $\mu \in \mathcal A$ satisfying the minimum-error equation
    \begin{equation}\label{eq:th:cone_gbpf_err}\tag{\thetheoremABC:objective}
        \mu = \mathrm{arg~min}_{\mu \in \mathcal A} \oint_{\mathrlap{S^{n-1}}} \; \mathrm d^{n-1} \hat \theta \int_{\mathrlap{\hat \theta^\perp}} \mathrm d^{n-1} r_\perp \; d_{\pm}(\hat \theta) \left(\left(m_{\pm} - A\mu\right)(\hat \theta, r_\perp) \right)^2 \, ,
    \end{equation}
    given by 
    \begin{equation}\label{eq:th:cone_gbpf_inv_formula}\tag{\thetheoremABC:BPC}
        \mu \! = \!  
        \underbrace{
        F_{\mathcal A}^{-1} \mathrm{diag}_k\left(\frac{|k|}{2 \pi \mathcal I[d_{\pm}](\hat k)}
            \right) F_{\mathcal A}
        }_{\text{volume deconvolution}} \!
            \underbrace{\Biggl\llbracket_{\mathrlap{\mathcal M}}^{\mathrlap{\mathcal A}} 
            p \mapsto \!\! \Biggl\llbracket_{\mathrlap{\mathbb R^n}}^{\mathrlap{\mathbb C}}
                 r \! \mapsto \!\! \int_{\mathrlap{\mathbb R^n}} \mathrm d^n x \,  \frac{
                    \chi(x) d_{\pm}(\hat \theta_{xr}) p(\hat \theta_{xr}, x)}{
                    \left( \frac 1 2 \int_{\mathbb R} \mathrm d s |s|^{n-1} \chi(r + s \hat \theta_{xr})\right)
                }
            \Biggr\rrbracket\Biggr\rrbracket}_{\text{weighted backprojection}} m_\pm \, ,
    \end{equation}
    where $\hat k = k/|k|$ and $\hat \theta_{xr}$ is the unit vector parallel to $(r - x)$.
    (Remark: geometrically, the expression $m_{\pm}(\hat \theta_{xr}, x)$ is the measurement corresponding to the undirected line connecting $x$ and $r$.)
\end{theoremABC}

\paragraph{Note on the symmetry assumption of theorem~\ref{th:cone_gbpf}}
The theorem has been deliberately weakened by assuming the symmetry $m(\hat \theta, r) = m(- \hat \theta, r)$. (Without loss of generality, we can then also assume that $d(\hat \theta) = d(- \hat \theta)$.) The symmetry holds if $m(\hat \theta, r)$ really does represent a line integral of an underlying function. With experimental measurement noise, the symmetry may hold only approximately. The reason this assumption has been made is that it has allowed us to incorporate a discretisation regularisation that is necessary for the inversion formula \eqref{eq:th:cone_gbpf_inv_formula} to be effectively applied as a computational formula. Further detail on this matter is given in the proof of theorem~\ref{th:cone_gbpf}, which is found in \S\ref{sec:proofs}. 

\paragraph{Note on the amenability assumption of theorem~\ref{th:cone_gbpf}}
Define the `relevant lines' as those lines intersecting the reconstruction support of the object and pointing in any direction $\hat \theta$ where $d_\pm(\hat \theta) \neq 0$. The amenability condition \eqref{eq:th:cone_gbpf_amenability} essentially states that relevant lines must each intersect the support of the source-point distribution $\chi$ at more than $1$ point in $\mathbb R^n$.

\section{Example applications of the results in \texorpdfstring{$3$}{3} dimensions}\label{sec:examples}
We state now some examples of how to apply theorems \ref{th:parallel_gbpf} and \ref{th:cone_gbpf} to construct explicit inversion formulae for various X-ray transmission experiments in 3-dimensional space. The spherical polar and cylindrical coordinate systems are illustrated in \ref{fig:ex_coords_spher} and \ref{fig:ex_coords_cylind} respectively.

\begin{figure}[!ht]
\centering%
\begin{subfigure}{0.5\linewidth}
\centering
\tdplotsetmaincoords{60}{120}
\begin{tikzpicture}[scale=4.0, opacity=1.0, line join=bevel, tdplot_main_coords, fill opacity=1.0, 
    vector/.style={-stealth,black,very thick},
	vector guide/.style={dashed,black,thick},
    angle/.style={black}]
    \def\rnum{0.75}
    \def\angrnum{0.55}
    \def\thetanum{65}
    \def\phinum{75}
	\pgfsetlinewidth{.2pt}
    
	\draw[color=black,thick,->] (0,0,0) -- (1,0,0) node[anchor=north east]{$x$};
	\draw[color=black,thick,->] (0,0,0) -- (0,1,0) node[anchor=north west]{$y$};
	\draw[color=black,thick,->] (0,0,0) -- (0,0,0.7) node[anchor=south]{$z$};
    
	\tdplotsetcoord{P}{\rnum}{\thetanum}{\phinum}
	\tdplotdrawarc[angle]{(0,0,0)}{\angrnum}{0}{\phinum}{anchor=north}{$\phi$}
	\tdplotsetthetaplanecoords{\phinum}
	\tdplotdrawarc[tdplot_rotated_coords,angle,opacity=1]{(0,0,0)}{\angrnum}{0}{\thetanum}{anchor=south west}{$\theta$}
	\draw[vector,opacity=1] (0,0,0) -- (P);
	\draw[vector guide] (0,0,0) -- (Pxy);
	\draw[vector guide] (Pxy) -- (P);
    \node[anchor=south east] at ($(0,0,0)!0.5!(P)$) {$r$};
    \node[anchor=south] at (P) {$(r,\theta,\phi)$};
\end{tikzpicture}
\caption{Spherical polar coordinate system, $(r, \theta, \phi)$.}
\label{fig:ex_coords_spher}
\end{subfigure}%
\hfill%
\begin{subfigure}{0.5\linewidth}
\centering
\tdplotsetmaincoords{60}{150}
\begin{tikzpicture}[scale=4.0, opacity=1.0, line join=bevel, tdplot_main_coords, fill opacity=1.0, 
    vector/.style={-stealth,black,very thick},
	vector guide/.style={dashed,black,thick},
    angle/.style={black}]
    \def\rnum{1.05}
    \def\angrnum{0.75}
    \def\thetanum{65}
    \def\phinum{35}
	\pgfsetlinewidth{.2pt}
    
	\draw[color=black,thick,->] (0,0,0) -- (1,0,0) node[anchor=north east]{$x$};
	\draw[color=black,thick,->] (0,0,0) -- (0,1,0) node[anchor=north west]{$y$};
	\draw[color=black,thick,->] (0,0,0) -- (0,0,0.7) node[anchor=south]{$z$};
    
	\tdplotsetcoord{P}{\rnum}{\thetanum}{\phinum}
	\tdplotdrawarc[angle]{(0,0,0)}{\angrnum}{0}{\phinum}{anchor=north east}{$\phi$}
	\draw[vector,opacity=1] (0,0,0) -- (P);
	\draw[color=black,-] (0,0,0) -- (Pxy);
    \node[anchor=north west] at ($(0,0,0)!0.5!(Pxy)$) {$\rho$};
    \draw[vector guide] (Pxy) -- (P);
    \node[anchor=south] at (P) {$(\rho,\phi,z)$};
\end{tikzpicture}
\caption{Cylindrical coordinate system, $(\rho, \phi, z)$.}
\label{fig:ex_coords_cylind}
\end{subfigure}
\caption{Spherical polar and cylindrical coordinate systems.}
\end{figure}

\subsection{Theorem~\ref{th:parallel_gbpf} / parallel-beam examples in \texorpdfstring{$3$}{3} dimensions}\label{sec:examples_parallel}
To assist in understanding how the first theorem (theorem~\ref{th:parallel_gbpf}) is applied, we describe two examples of constructing inversion formulae for parallel-beam experiments. Theorem~\ref{th:parallel_gbpf} describes two reconstruction formulae from parallel-beam data: the backprojection-convolution formula \eqref{eq:th:gbpf_inv_formula}, and the convolution-backprojection formula \eqref{eq:th:fbp_inv_formula}. The theorem requires us to choose a distribution $d$ on $S^2$, which determines both the data required by the reconstruction formulae and the error metric that will be minimised by the reconstruction. The reconstruction formulae in theorem~\ref{th:parallel_gbpf} use the data from all those directions $\hat \theta \in S^2$ for which $d(\hat \theta) \neq 0$. In this section we choose the distribution $d$---for simplicity and so that our results reduce to known methods in the literature---to be the same distribution that describes the directions in $S^2$ from which projections are taken. The density $d(\hat \theta)$ will correspond to how many projections are taken per unit solid angle in $S^2$, in the direction of $\hat \theta \in S^2$. Choosing $d$ to match the distribution of projections is a sensible choice because it apportions an equal weight in the error function per projection. 
\begin{figure}[!ht]
\centering%
\begin{subfigure}{0.5\linewidth}
\centering
\tdplotsetmaincoords{70}{135}
\begin{tikzpicture}[scale=3.5, opacity=0.0, line join=bevel, tdplot_main_coords, fill opacity=.5]
	\pgfsetlinewidth{.2pt}
    \tdplotsetpolarplotrange{90-5}{90+5}{0}{360}
	\tdplotsphericalsurfaceplot{73}{36}{0.5}{black}{black}%
		{\draw[color=black,thick,->] (0,0,0) -- (1,0,0) node[anchor=north east]{$x$};}%
		{\draw[color=black,thick,->] (0,0,0) -- (0,1,0) node[anchor=north west]{$y$};}%
		{\draw[color=black,thick,->] (0,0,0) -- (0,0,0.7) node[anchor=south]{$z$};}%
\end{tikzpicture}
\caption{Circular scan.}
\label{fig:ex_parallel_a}
\end{subfigure}%
\hfill%
\begin{subfigure}{0.5\linewidth}
\centering
\tdplotsetmaincoords{70}{135}
\begin{tikzpicture}[scale=3.5, opacity=0.0, line join=bevel, tdplot_main_coords, fill opacity=.5, 
    vector/.style={-stealth,black,very thick},
	vector guide/.style={dashed,black,thick},
    angle/.style={dashed, black, thick}]
    \def\R{0.5}
    \def\angR{0.85}
    \def\psinum{25}
    \def\phinum{0}
	\pgfsetlinewidth{.2pt}

    
    \tdplotsetpolarplotrange{90-\psinum}{90+\psinum}{0}{360}
	\tdplotsphericalsurfaceplot{72}{72}{\R}{black}{black}%
		{\draw[color=black,thick,->] (0,0,0) -- (1,0,0) node[anchor=north east]{$x$};}%
		{\draw[color=black,thick,->] (0,0,0) -- (0,1,0) node[anchor=north west]{$y$};}%
		{\draw[color=black,thick,->] (0,0,0) -- (0,0,0.7) node[anchor=south]{$z$};}%

	\tdplotsetcoord{P}{\angR}{90-\psinum}{\phinum}
	\tdplotsetthetaplanecoords{\phinum}
	\tdplotdrawarc[tdplot_rotated_coords,angle,opacity=1]{(0,0,0)}{\angR}{90}{90-\psinum}{anchor=south east}{$\psi$}
	\draw[vector,opacity=1] (0,0,0) -- (P);
\end{tikzpicture}
\caption{Spheroannular scan with $\psi = 25 \, \mathrm{deg}$.}
\label{fig:ex_parallel_b}
\end{subfigure}
\caption{Illustrations of the two example trajectories covered in the parallel-beam examples section \S\ref{sec:examples_parallel}. Each diagram illustrates a distribution on $S^2$ that specifies the density of projections taken at each direction in $S^2$. \ref{fig:ex_parallel_a} illustrates the distribution $d(\theta, \phi)$ in \eqref{eq:beam_dist_circ}. \ref{fig:ex_parallel_b} illustrates the distribution $d(\theta, \phi)$ in \eqref{eq:beam_dist_spheroannular}. The former is a limiting case of the latter, with $\psi \rightarrow 0$ as the total integrated density (i.e. number of projections) is held constant.}
\end{figure}

\subsubsection{Parallel-beam example 1: circular scan}\label{sec:ex_parallel_a}
A simple example application of theorem~\ref{th:parallel_gbpf} is to a circular scan in $3$ dimensions, with $c$ evenly-spaced parallel-beam projections taken. This is the most standard and basic of X-ray scans. The resulting reconstruction formula is not new; this example is included for pedagogical purposes. Since it will be familiar to many readers, this example serves as a demonstration that the theorem correctly reduces to that well-known result for the `filtered backprojection' (i.e. convolution-backprojection) formula for reconstruction. 

Let $(\theta, \phi) \in S^2$ be spherical polar coordinates, with $\theta$ the polar coordinate and $\phi$ the azimuthal. The coordinate system is illustrated in fig.~\ref{fig:ex_coords_spher}, with $r=1$. The distribution of beam directions $d(\hat \theta)$ is
\begin{equation}
    d(\theta, \phi) = \delta\left(\theta - \frac{\pi}{2}\right) \sum_{m=1}^c \delta\left( \phi - \frac{m}{c} 2\pi \right) \, .
\end{equation}
It is easy to see that this data is insufficient to exactly reconstruct a continuous volume, because it is easy to choose a great-circle (or `equatorial band') of $S^2$ which does not pass through the support of $d(\theta, \phi)$. Stated more precisely: the Funk transform $\mathcal I[d](\hat \theta)$, which integrates the distribution $d$ around the great-circle orthogonal to $\hat \theta \in S^2$, is equal to $0$ for one or more choices of $\hat \theta \in S^2$. This violates the necessary and sufficient condition for data sufficiency expressed in eq.~\eqref{eq:th:gbpf_existence} of theorem~\ref{th:parallel_gbpf}, which requires that $\mathcal I[d](\hat \theta) \neq 0$. As is standard practice in the construction of analytic inversion methods, we make the approximation that the source points comprise a continuous distribution instead of a discrete one:
\begin{equation}\label{eq:beam_dist_circ}
    d(\theta, \phi) = \delta\left(\theta - \frac{\pi}{2}\right) \frac{c}{2 \pi} \, .
\end{equation}
This distribution is illustrated in fig.~\ref{fig:ex_parallel_a}. This distribution satisfies the data sufficiency requirement that $\mathcal I[d](\hat \theta) \neq 0$.

The Funk transform of an arbitrary nonzero vector $k = (k_x, k_y, k_z)$ can be expressed in closed form as:\footnote{
Whether the Funk transform has closed form is an essential aspect in the computational application of the convolution/backprojection-style methods. If it did not have closed form, then computing the deconvolution would be unwieldy or impractical because it would require a separate numerical integration for each $\hat k$.
} 
\begin{equation}
    \mathcal I[d](\hat k) 
    = \oint_{\mathrlap{\{\hat \theta \in S^{n-1} : \hat \theta \cdot \hat k = 0\}}} \; \mathrm d^{n-2} \hat \theta \;\;\;\, d(\hat \theta) 
    = \frac{c}{\pi} \frac{|k|}{\sqrt{k_x^2 + k_y^2}} \, .
\end{equation}
Per theorem~\ref{th:parallel_gbpf}, the formula for the deconvolution step simplifies viz.
\begin{equation}\label{eq:ex_parallel_01}
    F_{}^{-1} \mathrm{diag}_k \left( 
        \frac{|k|}{2 \pi \mathcal I[d](\hat k)}
    \right) F_{} =
    F_{}^{-1} \mathrm{diag}_k \left( 
        \frac{\sqrt{k_x^2 + k_y^2}}{2c}
    \right) F_{} \, ,
\end{equation}
where $F = F_{\mathcal A}$ in the case of the backprojection-convolution formula \eqref{eq:th:gbpf_inv_formula}, or $F = F_{\mathcal M}$ in the case of the convolution-backprojection formula \eqref{eq:th:fbp_inv_formula}.

Per theorem~\ref{th:parallel_gbpf}, the formula for the backprojection step simplifies viz.
\begin{equation}
\Bigg\llbracket_{\mathcal M}^{\mathcal A} p \mapsto
\Bigl\llbracket_{\mathbb R^n}^{\mathbb C}
                r \mapsto \oint_{\mathrlap{S^{2}}} \; \mathrm d^{2} \hat \theta \; d(\hat \theta) p(\hat \theta, r) \Bigr\rrbracket \Bigg\rrbracket
= 
\Bigg\llbracket_{\mathcal M}^{\mathcal A} p \mapsto
\Bigl\llbracket_{\mathbb R^n}^{\mathbb C}
                r \mapsto \frac{c}{2 \pi} \int_{\mathrlap{0}}^{\mathrlap{2 \pi}} \mathrm d \Phi \; p\left((\tfrac{\pi}{2}, \Phi), r\right) \Bigr\rrbracket \Bigg\rrbracket \, .
\end{equation}
Since we approximated a discrete angular distribution by a continuous one, it is appropriate to next write the integration over $\Phi$ as a summation over the actual angles $\Phi_i$ instead of the idealised continuous distribution, so that the equation may actually be used with discrete data. This is done by replacing the integration with a summation viz.
\begin{equation}
    \frac{c}{2\pi} \int_{\mathrlap{0}}^{\mathrlap{2\pi}} \mathrm d \Phi \rightarrow \sum_{\Phi_m} \, .
\end{equation}
The simplified formula for the backprojection step is thus
\begin{equation}\label{eq:ex_parallel_02}
\Bigg\llbracket_{\mathcal M}^{\mathcal A} p \mapsto
\Bigl\llbracket_{\mathbb R^n}^{\mathbb C}
                r \mapsto \sum_{i=1}^c p\left((\tfrac{\pi}{2}, \phi_i), r\right) \Bigr\rrbracket \Bigg\rrbracket \, .
\end{equation}

The backprojection-convolution reconstruction formula \eqref{eq:th:gbpf_inv_formula} is obtained by performing the backprojection step \eqref{eq:ex_parallel_02} prior to the convolution step \eqref{eq:ex_parallel_01}. This gives the reconstruction formula:
\begin{equation}
    \mu = 
    \underbrace{
    F_{\mathcal A}^{-1} \mathrm{diag}_k\left(\frac{\sqrt{k_x^2 + k_y^2}}{2c } \right)  F_{\mathcal A}
    }_{\text{volume deconvolution}} \;
        \underbrace{\Bigl\llbracket_{\mathbb R^n}^{\mathbb C}
                r \mapsto \sum_{i=1}^c m\left((\tfrac{\pi}{2}, \phi_i), r\right) \Bigr\rrbracket}_{\text{weighted backprojection of $m$}} \, .
\end{equation}
This is a way of phrasing a standard result for `backprojection-filtration' reconstruction from circular parallel-beam scanning data.

The convolution-backprojection reconstruction formula \eqref{eq:th:fbp_inv_formula} is obtained by performing the convolution step \eqref{eq:ex_parallel_01} prior to the backprojection step \eqref{eq:ex_parallel_02}. This gives the reconstruction formula:
\begin{equation}
    \mu = 
    \underbrace{
    \Biggl\llbracket_{\mathcal M}^{\mathcal A} p \mapsto 
    \Bigl\llbracket_{\mathbb R^n}^{\mathbb C}
                r \mapsto \sum_{i=1}^c p\left((\tfrac{\pi}{2}, \phi_i), r\right) \Bigr\rrbracket \Biggr\rrbracket
        }_{\text{weighted backprojection}} \underbrace{
    F_{\mathcal M}^{-1} \mathrm{diag}_k\left({\tfrac 1 {2c} \sqrt{k_x^2 + k_y^2}} \right)  F_{\mathcal M}
    }_{\text{projection deconvolution}} m \, .
\end{equation}
This is a reproduction of the well-known `filtered backprojection' method for a circular scan with a parallel-beam.

In practice, the convolution-backprojection method is practical to perform for parallel beam data, since the convolution $F_{\mathcal M}$ is a 2D convolution on each separate measurement set $m(\hat \theta, \cdot)$, and this corresponds with the data format that is actually collected by 2D X-ray detectors in a parallel-beam experiment. The same cannot be said for cone-beam reconstruction.

\subsubsection{Parallel-beam example 2: spheroannular scan}\label{sec:spheroannular}
This example is a generalisation of the circular scan.
Spherical polar coordinates are illustrated in fig.~\ref{fig:ex_coords_spher}.

Let $(\theta, \phi) \in S^2$ be spherical polar coordinates, with $\theta$ the polar coordinate and $\phi$ the azimuthal. The distribution of beam directions $d(\hat \theta)$ we assume is a uniform distribution around the equatorial band $|\theta - \pi/2| < \psi$, with $\psi \leq \pi/2$. The total area of that band is $\int_{\pi/2-\psi}^{\pi/2+\psi} \sin \theta = 4 \pi \sin \psi$. 
In experimental practice, this may be approximated with a discrete distribution. 
However, to facilitate the application of the continuous theory, we assume the continuous density:
\begin{equation}\label{eq:beam_dist_spheroannular}
    d(\theta, \phi) = \begin{cases}
        c / (4 \pi \sin \psi) &\text{if} \quad |\theta - \pi/2| < \psi\\
        0 &\text{otherwise}
    \end{cases} \, .
\end{equation}
This distribution is illustrated in fig.~\ref{fig:ex_parallel_b}. 

The Funk transform of an arbitrary nonzero vector $k = (k_r, k_\theta, k_\phi)$ can be expressed in closed form as:
\begin{equation}
    \mathcal I[d](\hat k) = \oint_{\mathrlap{\{\hat \theta \in S^{2} : \hat \theta \cdot \hat k = 0\}}} \; \mathrm d^{2} \hat \theta \;\;\;\, d(\hat \theta) 
    = \frac{c}{\pi \sin \psi} 
    \arcsin \left( \frac{\sin \psi}{\max\left\lbrace
            \sin \psi , \sin k_\theta
        \right\rbrace} \right) 
    \, .
\end{equation}

Per theorem~\ref{th:parallel_gbpf}, the formula for the deconvolution step simplifies viz.
\begin{equation}\label{eq:ex_parallel_spherannular_01}
    F_{}^{-1} \mathrm{diag}_k \left( 
        \frac{|k|}{2 \pi \mathcal I[d](\hat k)}
    \right) F_{} =
    F_{}^{-1} \mathrm{diag}_k \left( 
        \frac{\sin \psi}{2c}
        \frac{|k|}{
            \arcsin \left( \frac{\sin \psi}{\max\left\lbrace
            \sin \psi , \sin k_\theta
        \right\rbrace} \right)
        }
    \right) F_{} \, ,
\end{equation}
where $F = F_{\mathcal A}$ in the case of the backprojection-convolution formula \eqref{eq:th:gbpf_inv_formula}, or $F = F_{\mathcal M}$ in the case of the convolution-backprojection formula \eqref{eq:th:fbp_inv_formula}.

Per theorem~\ref{th:parallel_gbpf}, the formula for the backprojection step is 
\begin{equation}\label{eq:ex_parallel_spherannular_02}
\Bigg\llbracket_{\mathcal M}^{\mathcal A} p \mapsto
\Bigl\llbracket_{\mathbb R^n}^{\mathbb C}
                r \mapsto \oint_{\mathrlap{S^{2}}} \; \mathrm d^{2} \hat \theta \; d(\hat \theta) p(\hat \theta, r) \Bigr\rrbracket \Bigg\rrbracket
\, .
\end{equation}

The backprojection-convolution reconstruction formula \eqref{eq:th:gbpf_inv_formula} is obtained by performing the backprojection step \eqref{eq:ex_parallel_spherannular_02} prior to the convolution step \eqref{eq:ex_parallel_spherannular_01}. This gives the reconstruction formula:
\begin{equation}
    \mu = 
    \underbrace{
    F_{\mathcal A}^{-1} 
    \mathrm{diag}_k\left(
        \frac{\sin \psi}{2c}
        \frac{|k|}{
            \arcsin \left( \frac{\sin \psi}{\max\left\lbrace
            \sin \psi , \sin k_\theta
        \right\rbrace} \right)
        }
    \right)  F_{\mathcal A}
    }_{\text{volume deconvolution}} \;
        \underbrace{
        \Bigl\llbracket_{\mathbb R^n}^{\mathbb C}
                r \mapsto \oint_{\mathrlap{S^{2}}} \; \mathrm d^{2} \hat \theta \; d(\hat \theta) m(\hat \theta, r) \Bigr\rrbracket
        }_{\text{weighted backprojection of $m$}} \, .
\end{equation}

The convolution-backprojection reconstruction formula \eqref{eq:th:fbp_inv_formula} is obtained by performing the convolution step \eqref{eq:ex_parallel_spherannular_01} prior to the backprojection step \eqref{eq:ex_parallel_spherannular_02}. This gives the reconstruction formula:
\begin{equation}
    \mu = 
    \underbrace{
    \Biggl\llbracket_{\mathrlap{\mathcal M}}^{\mathrlap{\mathcal A}} p \mapsto 
        \Bigl\llbracket_{\mathrlap{\mathbb R^n}}^{\mathrlap{\mathbb C}}
                r \mapsto \oint_{\mathrlap{S^{2}}} \; \mathrm d^{2} \hat \theta \; d(\hat \theta) p(\hat \theta, r) \Bigr\rrbracket
        }_{\text{weighted backprojection}} \underbrace{
    \!F_{\mathcal M}^{-1} \mathrm{diag}_k\left(\!\!
        \frac{\sin \psi}{2c}
        \frac{|k|}{
            \arcsin \left( \frac{\sin \psi}{\max\left\lbrace
            \sin \psi , \sin k_\theta
        \right\rbrace} \right)
        }
    \!\!\right) F_{\mathcal M}
    }_{\text{projection deconvolution}} m \, .
\end{equation}

These results are equivalent to the methods published by \emph{Colsher} in 1980 \cite{colsher1980fully} (see also \cite{defrise2005image} for a review of `filtered backprojection' methods in Positron Emission Tomography).

\subsection{Theorem~\ref{th:cone_gbpf} / cone-beam examples in \texorpdfstring{$3$}{3} dimensions}
To assist in understanding how theorem~\ref{th:cone_gbpf} is applied, we describe two examples of constructing backprojection-convolution inversion formulae for cone-beam experiments. 
These formulae provide the first (to our knowledge) analytic inversion formulae for cone-beam X-ray experiments where the locus of source points, $X \subset \mathbb R^3$, is multidimensional. Such experiments are of active interest \cite{myers2016rapidly, kingston2018space, bauer2021, grewar2024preprint} and have not previously been approachable with direct methods of reconstruction. 

The reconstruction formula is given in \eqref{eq:th:cone_gbpf_inv_formula}. The imaging geometry is specified by the choice of distribution $d$ on $S^2$ (this limits the directions of measurement lines required by the reconstruction) and by the source-point density function $\chi$ (this specifies the location and density of X-ray source points in 3D space). To apply the formula efficiently, two integrals should be evaluated in closed form, if possible: (1) the Funk transform $\mathcal I[d](\hat k)$, and (2) the line integrals of source-point density:
\begin{equation}\label{eq:ex_cone_lineintegral}
    \frac{1}{2} \int_{-\infty}^{\infty} \mathrm d s \,  s^2 \chi(r + s \hat \theta_{xr}) \, ,
\end{equation}
for each $r \in \mathbb R^3$ and $x \in X \subset \mathbb R^3$, and where $\hat \theta_{xr}$ is the unit-length vector parallel to $(r-x)$. 

\begin{figure}[!ht]
\centering%
\begin{subfigure}{0.5\linewidth}
\centering
\tdplotsetmaincoords{70}{135}
\begin{tikzpicture}[scale=3.5, opacity=0.1, line join=bevel, tdplot_main_coords, fill opacity=.125, 
    vector/.style={-stealth,black,very thick},
	vector guide/.style={dashed,black,thick},
    angle/.style={dashed, black, thick}]
    \def\R{0.5}
    \def\angR{0.85}
    \def\detPhi{170}
    \def\detPhiSpan{70}
    \def\detThetaSpan{70}
    \def\phinum{0}
	\pgfsetlinewidth{.2pt}

    
    \tdplotsetpolarplotrange{0}{180}{0}{360}
	\tdplotsphericalsurfaceplot{36}{36}{\R}{black}{black}%
		{\draw[color=black,thick,->] (0,0,0) -- (0.9,0,0) node[anchor=north east]{};}%
		{\draw[color=black,thick,->] (0,0,0) -- (0,0.9,0) node[anchor=north west]{};}%
		{\draw[color=black,thick,->] (0,0,0) -- (0,0,0.7) node[anchor=south]{};}%
    
	\tdplotsetcoord{S}{\R}{90}{\detPhi}
    \tdplotsetcoord{D1}{\R*1.2}{90+\detThetaSpan/2}{\detPhi-180+\detPhiSpan/2}
    \tdplotsetcoord{D2}{\R*1.2}{90+\detThetaSpan/2}{\detPhi-180-\detPhiSpan/2}
    \tdplotsetcoord{D3}{\R*1.2}{90-\detThetaSpan/2}{\detPhi-180-\detPhiSpan/2}
    \tdplotsetcoord{D4}{\R*1.2}{90-\detThetaSpan/2}{\detPhi-180+\detPhiSpan/2}
    \coordinate[black, fill opacity=1, circle, fill, inner sep={1.5pt}, pin={[pin edge={black, draw opacity=1, -}, opacity=1] 0:{$x$}}] (Scopy) at (S);
    \draw[color=black,opacity=1] (S) -- (D1);
    \draw[color=black,opacity=1] (S) -- (D2);
    \draw[color=black,opacity=1] (S) -- (D3);
    \draw[color=black,opacity=1] (S) -- (D4);
    \draw[color=black,opacity=0, fill opacity=0.1, fill=black] (S) -- (D3) -- (D4);
    \draw[color=black,opacity=0, fill opacity=0.1, fill=black] (S) -- (D2) -- (D3);
    \draw[color=black,opacity=0, fill opacity=0.1, fill=black] (S) -- (D1) -- (D2);
    \draw[color=black,opacity=0, fill opacity=0.1, fill=black] (S) -- (D4) -- (D1);
    \draw[color=black,opacity=1,fill=gray] (D1) -- (D2) -- (D3) -- (D4) -- (D1);
    \draw[transform canvas={xshift=-0.5mm, yshift=-0.75mm}, color=black,opacity=1,<->] (D2) -- (D1) node[midway, anchor=north east]{$W$};
    \draw[transform canvas={xshift=-1mm, yshift=0mm}, color=black,opacity=1,<->] (D2) -- (D3) node[midway, anchor=east]{$W$};
    
\end{tikzpicture}
\caption{}
\label{fig:ex_cone_spher_a}
\end{subfigure}%
\hfill%
\begin{subfigure}{0.5\linewidth}
\centering
\begin{tikzpicture}[
    my angle/.style={draw, <->, angle eccentricity=1.3, angle radius=9mm},
    extended line/.style={shorten >=-#1,shorten <=-#1},
    extended line/.default=0.5cm]
    ]
    \coordinate[] (l) at (-30:2cm);
    \draw ([shift=(l)] 0,0) arc (-30:210:2cm);
    \coordinate[circle, fill, inner sep={1.5pt}, pin=190:{$O$}] (O)  at (0,0);
    \coordinate[circle, fill, inner sep=1.5pt, pin= 0:{$x$}]  (S)  at (2,0);
    \coordinate[] (e) at (120:2);
    \coordinate[circle, fill, inner sep={1.5pt}, pin= 200:{$r$}] (p) at ($(S)!.8!(e)$);
    \coordinate[] (c) at ($(S)!.5!(e)$);
    \coordinate[gray, circle, fill, inner sep={1.5pt}] (D) at ($(S)!1.5!(e)$);
    \path let \p1 = (D) in coordinate (L) at (\x1,0);

    \draw[transform canvas={xshift=0mm, yshift=0mm}, <->] (S) -- (O) node[midway, below]{$R$};
    \draw[transform canvas={xshift=0mm, yshift=-6mm}, <->] let  
      \p1 = (S),
      \p2 = ($(L)$)
    in
      (\x1,\y1) -- (\x2,\y2)
      node[midway,below]{$L$};
    \draw[<->] (O) -- (p) node[midway, left]{$|r|$};

    \draw pic["$\varphi$",draw=black,<->,angle eccentricity=1.2,angle radius=1cm] {angle=p--S--O};

    \draw[thick, gray] let 
      \p1 = (D),
      \p2 = (l)
    in 
      (\x1,\y2) node[below, gray] {} -- (\x1,3);
    \draw[gray] let \p1=(D) in (\x1,0.3) node[above, rotate=90] {detector};

    \draw[thick, dashed] (S) -- (D);
\end{tikzpicture}
\vspace*{-2.5mm}
\caption{}
\label{fig:ex_cone_spher_b}
\end{subfigure}
\caption{The spherical 3D cone-beam imaging geometry in \S\ref{sec:ex_cone_spher}. \textbf{(a):} The source locus $X$ is a sphere of radius $R$, centred at the coordinate origin. \textbf{(b):} a 2D cross-section of the plane containing the origin $O$, the source point $x \in X$, and the volume point $r \in \mathbb R^3$.}
\label{fig:ex_cone_spher}
\end{figure}
\begin{figure}[!ht]
\centering%
\begin{subfigure}{0.3\linewidth}
\centering
\resizebox{1.0\columnwidth}{!}{%
\tdplotsetmaincoords{70}{135}
\begin{tikzpicture}[scale=2.5, opacity=0.1, line join=bevel, tdplot_main_coords, fill opacity=.125, 
    vector/.style={-stealth,black,very thick},
	vector guide/.style={dashed,black,thick},
    angle/.style={dashed, black, thick}]
    \def\R{0.5}
    \def\angR{0.85}
    \def\detPhi{170}
    \def\detPhiSpan{90}
    \def\detThetaSpan{50}
    \def\phinum{0}
	\pgfsetlinewidth{.2pt}

    
    \tdplotsetpolarplotrange{32.5}{147.5}{0}{360}
	\tdplotsphericalsurfaceplot{36}{36}{\R/sin(\tdplottheta)}{black}{black}%
		{\draw[color=black,thick,->] (0,0,0) -- (0.8,0,0) node[anchor=north east]{};}%
		{\draw[color=black,thick,->] (0,0,0) -- (0,0.8,0) node[anchor=north west]{};}%
		{\draw[color=black,thick,->] (0,0,0) -- (0,0,0.8) node[anchor=south]{$z$};}%

	\tdplotsetcoord{S}{\R/sin(90)}{90}{\detPhi}
    \tdplotsetcoord{D1}{\R*1.3}{90+\detThetaSpan/2}{\detPhi-180+\detPhiSpan/2}
    \tdplotsetcoord{D2}{\R*1.3}{90+\detThetaSpan/2}{\detPhi-180-\detPhiSpan/2}
    \tdplotsetcoord{D3}{\R*1.3}{90-\detThetaSpan/2}{\detPhi-180-\detPhiSpan/2}
    \tdplotsetcoord{D4}{\R*1.3}{90-\detThetaSpan/2}{\detPhi-180+\detPhiSpan/2}
    \coordinate[black, fill opacity=1, circle, fill, inner sep={1.5pt}, pin={[pin edge={black, draw opacity=1, -}, opacity=1] 70:{$x$}}] (Scopy) at (S);
    \draw[color=black,opacity=1] (S) -- (D1);
    \draw[color=black,opacity=1] (S) -- (D2);
    \draw[color=black,opacity=1] (S) -- (D3);
    \draw[color=black,opacity=1] (S) -- (D4);
    \draw[color=black,opacity=0, fill opacity=0.1, fill=black] (S) -- (D3) -- (D4);
    \draw[color=black,opacity=0, fill opacity=0.1, fill=black] (S) -- (D2) -- (D3);
    \draw[color=black,opacity=0, fill opacity=0.1, fill=black] (S) -- (D1) -- (D2);
    \draw[color=black,opacity=0, fill opacity=0.1, fill=black] (S) -- (D4) -- (D1);
    \draw[color=black,opacity=1,fill=gray] (D1) -- (D2) -- (D3) -- (D4) -- (D1);
    \draw[transform canvas={xshift=-0.5mm, yshift=-0.75mm}, color=black,opacity=1,<->] (D2) -- (D1) node[midway, anchor=north east]{$W$};
    \draw[transform canvas={xshift=-1mm, yshift=0mm}, color=black,opacity=1,<->] (D2) -- (D3) node[midway, anchor=east]{$H$};
    
\end{tikzpicture}
}
\caption{source cylinder}
\label{fig:ex_cone_cyl_a}
\end{subfigure}%
\hfill%
\begin{subfigure}{0.38\linewidth}
\centering
\resizebox{1.0\columnwidth}{!}{%
\begin{tikzpicture}[
    my angle/.style={draw, <->, angle eccentricity=1.3, angle radius=9mm},
    extended line/.style={shorten >=-#1,shorten <=-#1},
    extended line/.default=0.5cm]
    ]
    \pgfmathsetmacro{\W}{4.5}
    
    \coordinate[] (l) at (-90:2cm);
    \draw ([shift=(l)] 0,0) arc (-90:270:2cm);
    \coordinate[circle, fill, inner sep={1.5pt}, pin=190:{$O$}] (O)  at (0,0);
    \coordinate[circle, fill, inner sep=1.5pt, pin= 115:{$x$}]  (S)  at (2,0);
    \coordinate[] (e) at (135:2);
    \coordinate[circle, fill, inner sep={1.5pt}, pin= 200:{$r$}] (p) at ($(S)!.8!(e)$);
    \coordinate[] (c) at ($(S)!.5!(e)$);
    \coordinate[gray, circle, fill, inner sep={1.5pt}] (D) at ($(S)!1.25!(e)$);
    \path let \p1 = (D) in coordinate (L) at (\x1,0);

    \draw[transform canvas={xshift=0mm, yshift=0mm}, <->] (S) -- (O) node[midway, below]{$R$};
    \draw[transform canvas={xshift=0mm, yshift=-6mm}, <->] let  
      \p1 = (S),
      \p2 = ($(L)$)
    in
      (\x1,\y1) -- (\x2,\y2)
      node[midway,below]{$L$};
    \draw[<->] (O) -- (p) node[midway, left]{$r_\rho$};

    \draw pic["$\theta_h$",draw=black,<->,angle eccentricity=1.2,angle radius=1.25cm] {angle=p--S--O};

    \draw[thick, gray] let 
      \p1 = (D),
      \p2 = (l)
    in 
      (\x1,\W/2) node[below, gray] {} -- (\x1,-\W/2);
    \draw[gray] let \p1=(D) in (\x1,0.3) node[above, rotate=90] {detector};
    \draw[<->] let 
      \p1 = (D),
      \p2 = (l)
    in 
      (\x1-0.5cm,\W/2) node[below, gray] {} -- (\x1-0.5cm,-\W/2)
      node[midway,left]{$W$};

    \draw[thick, dashed] (S) -- (D);
\end{tikzpicture}
\vspace*{-2.5mm}
}
\caption{top-down view showing $\theta_h$}
\label{fig:ex_cone_cyl_b}
\end{subfigure}
\hfill%
\begin{subfigure}{0.31\linewidth}
\centering
\resizebox{1.0\columnwidth}{!}{%
\begin{tikzpicture}[
    my angle/.style={draw, <->, angle eccentricity=1.3, angle radius=9mm},
    extended line/.style={shorten >=-#1,shorten <=-#1},
    extended line/.default=0.5cm]
    ]

    \pgfmathsetmacro{\twoR}{9}
    \pgfmathsetmacro{\propofr}{0.8}
    
    \pgfmathsetmacro{\lenxr}{\twoR*sqrt((1+cos(135))/2)}
    
    \coordinate[circle, fill, inner sep=1.5pt, pin= 115:{$x$}] (S) at (\lenxr/2,0);
    \coordinate[] (T) at (-\lenxr/2,0);
    \coordinate[] (U) at (\lenxr/2,1);
    
    \coordinate[circle, fill, inner sep=1.5pt, pin= 60:{$r$}] (r) at (\lenxr/2 - \lenxr*\propofr, \twoR/2*0.2);

    \draw[thick, black, -] let 
      \p1 = (S),
      \p2 = (T)
    in 
      (\x1,-1.5) -- (\x1,2)
      (\x2,-1.5) -- (\x2,2);

    \draw pic["$\theta_{xr}$",draw=black,<->,angle eccentricity=1.2,angle radius=1.25cm] {angle=U--S--r};
    
    %
    \draw[transform canvas={xshift=0mm, yshift=-2mm}, <->] (S) -- (T) node[midway, below]{$2 R \sqrt{\tfrac{1}{2} - \tfrac{1}{2} \cos(2 \theta_h)}$};
    \draw[black, ->, dashed] (S) -- (r);
    
\end{tikzpicture}
}
\caption{vertical slice showing $\theta_{xr}$}
\label{fig:ex_cone_cyl_c}
\end{subfigure}
\caption{The cylindrical 3D cone-beam imaging geometry in \S\ref{sec:ex_cone_cylinder}. \textbf{(a):} The source locus $X$ is a cylinder of radius $R$, of infinite height, with its axis intersecting the coordinate origin. \textbf{(b):} a top-down view of the cylinder. The angle $\theta_h \in [0, \pi]$ is measured between the top-down 2D projection of the coordinate origin $O$, the source point $x \in X$, and the volume point $r \in \mathbb R^3$. The width of the detector is $W$, and the height of the detector (not visible) is $H$. \textbf{(c):} a vertical slice through the cylinder, along the plane parallel to the $z$ axis and containing the line through $x$ and $r$. The polar angle $\theta_{xr} \in [0, \pi]$ is visible in this side view.}
\label{fig:ex_cone_cyl}
\end{figure}

\subsubsection{Cone-beam example 1: spherical scan with finite detector}\label{sec:ex_cone_spher}
Consider a scan with a spherical source locus of radius $R > 0$, as approximated by the discrete trajectory in \cite{bauer2021}. An example illustration of a discretely sampled sphere is in fig.~\ref{fig:trajectories_lds_sphere}. A flat square detector of height and width $W$ is oriented opposite the source point at a distance $L > R$, so that measurement lines passing through the sphere are captured. The imaging geometry is depicted in fig.~\ref{fig:ex_cone_spher}. The source locus $X$ is the sphere expressed by: 
\begin{equation}
    X = \{ R \hat \theta : \hat \theta \in S^{2} \} \subset \mathbb R^3 \, ,
\end{equation}
and is illustrated in fig.~\ref{fig:ex_cone_spher_a}.
We assume a uniform 2-dimensional density of $c/(4 \pi R^2)$ source points per unit area on this sphere (for a total of $c$ projections). This corresponds to a 3D source-point density function $\chi$ which is given in spherical polar coordinates by:
\begin{equation}
    \chi ( (r, \theta, \phi) ) = \frac{c}{4 \pi R^2} \delta(r - R) \, .
\end{equation}

We have stipulated the source-point distribution $\chi$. In order to apply theorem~\ref{th:cone_gbpf} for reconstruction, we next need to establish a reconstruction support $V \subset \mathbb R^3$ and a distribution $d$ on $S^2$ that together satisfy both the amenability condition \eqref{eq:th:cone_gbpf_amenability} and sufficiency condition \eqref{eq:th:cone_gbpf_existence}. It is obvious from the spherical symmetry of the experiment that the reconstruction support ought to be a sphere centred at the origin of coordinates, and that $d_\pm(\hat \theta)$ ought to be independent of $\hat \theta$ to maintain symmetry. An overall factor on $d_\pm(\hat \theta)$ makes no difference to the reconstruction formula, so we choose $d_\pm(\hat \theta) = 1$, which clearly satisfies the sufficiency condition ($\mathcal I[d](\hat \theta) \neq 0$) because $\mathcal I[d](\hat \theta) = 2 \pi$ for all $\hat \theta \in S^2$. The amenability condition requires that `relevant lines' (lines pointing in any direction $\hat \theta$ where $d(\hat \theta) \neq 0$, and which intersect the reconstruction support) intersect the source-point distribution $\chi$ in at least two locations, and that the measurement data should be available for these lines. Since $d(\hat \theta) = 1$, this requires that all lines intersecting $V$ must intersect the sphere of radius $R$ in at least two locations. That implies that the radius of the reconstruction support is less than $R$. Also, the measurement data along `relevant lines' should actually be available, so we must abide the sharper constraint imposed by the size of the rectangular detector. We thus choose $V$ as the sphere of radius $R/\sqrt{1 + 4L^2/W^2}$ and $d(\hat \theta) = 1$.

Now that $V$ and $d$ have been chosen to satisfy the amenability and sufficiency conditions, the next step is to compute the Funk transform analytically (which we have already done: $\mathcal I[d](\hat \theta) = 2 \pi$) and the integral \eqref{eq:ex_cone_lineintegral} in closed form. 
Denote the origin of coordinates by $O$. Now select a source point $x \in X$ and volume point $r \in \mathbb R^3$. Define the angle $\varphi = \angle O x r$, illustrated in fig.~\ref{fig:ex_cone_spher_b}.
Then (without proof):
\begin{equation}
    \frac 1 2 \int_{-\infty}^{\infty} \mathrm d s \, s^2 \chi(r+s\hat \theta_{xr}) = \frac{c}{4 \pi} \frac{\cos(2 \varphi) + (|r|/R)^2}{|\cos \varphi|} \, ,
\end{equation}
where $\hat \theta_{xr}$ is the unit vector in the direction of $(r - x)$. 

The backprojection-convolution formula \eqref{eq:th:cone_gbpf_inv_formula} may now be applied to recover the minimum-error solution to the attenuating volume described in \eqref{eq:th:cone_gbpf_err}.

\subsubsection{Cone-beam example 2: cylindrical scan with finite detector}\label{sec:ex_cone_cylinder}
Consider a scan with a cylindrical source locus of radius $R > 0$, as approximated by the discrete trajectory in \cite{kingston2018space}. An example illustration of a discretely sampled cylinder is in fig.~\ref{fig:trajectories_sft}. We assume a rectangular detector of width $W$ and height $H$, at a distance $L$ from the source point. The imaging geometry is depicted in fig.~\ref{fig:ex_cone_cyl}.
We assume that the cylinder has infinite height (in practice, measurements become zero everywhere when the X-ray source is placed far above or below the reconstruction support, so this assumption is practically valid).
We use cylindrical coordinates $(\rho, \phi, z)$ for $\mathbb R^3$, illustrated in fig.~\ref{fig:ex_coords_cylind}. The source locus is
\begin{equation}
    X = \left\lbrace (\rho,\phi,z) \in \mathbb R^3 : \rho = R \right\rbrace \, ,
\end{equation}
and is illustrated in fig.~\ref{fig:ex_cone_cyl_a}.
We assume a uniform 2-dimensional density of $C$ source points per unit area on this cylinder. So, the distribution $\chi$ of source points in $\mathbb R^3$ is given by
\begin{equation}
    \chi\left((\rho,\phi,z)\right) = C \, \delta\left( \rho - R \right) \, .
\end{equation}

We have stipulated the source-point distribution $\chi$. In order to apply theorem~\ref{th:cone_gbpf} for reconstruction, we next need to establish a reconstruction support $V \subset \mathbb R^3$ and a distribution $d_\pm$ on $S^2$ that together satisfy both the amenability condition \eqref{eq:th:cone_gbpf_amenability} and sufficiency condition \eqref{eq:th:cone_gbpf_existence}. It is obvious from the cylindrical symmetry of the experiment that the reconstruction support ought to be a cylinder, and that $d_\pm(\hat \theta) \equiv d_\pm(\theta, \phi)$ ought depend only on $\theta$. The amenability condition requires that `relevant lines' (lines pointing in any direction $\hat \theta$ where $d_\pm(\hat \theta) \neq 0$, and which intersect the reconstruction support) intersect the source-point distribution $\chi$ in at least two locations, and that the measurement data should be available for these lines. The maximum allowable radius of $V$ is determined by the width of the detector to be $R/\sqrt{1+4L^2/W^2}$. With that choice of $V$, the maximum allowable support of $d_\pm(\hat \theta)$ is determined as the range of angles $\hat \theta \in S^2$ over which measurement data is available for all lines intersecting $V$ and pointing in the direction $\hat \theta$. Setting $d_\pm(\hat \theta) = 1$ over that maximal support, and $0$ elsewhere, we find
\begin{equation}
    d_{\pm}((\theta, \phi)) = \begin{cases}
        1 & \text{if} \quad |\theta - \pi/2| < \Omega_v/2 \\
        0 & \text{otherwise} 
    \end{cases} \, ,
    \quad \Omega_v = 2 \arctan \left( \frac{H/2}{\sqrt{L^2 + (W/2)^2}} \right) \, .
\end{equation}

Now that $V$ and $d$ have been chosen to satisfy the amenability and sufficiency conditions, the next step is to compute the Funk transform analytically and the integral \eqref{eq:ex_cone_lineintegral} in closed form. The Funk transform is the same as given earlier in \S\ref{sec:spheroannular}, because the weighting of the cone-beam backprojection emulates the spheroannular parallel-beam backprojection with $\psi = \Omega_v/2$. The Funk transform is
\begin{equation}
    \mathcal I[d_{\pm}](\hat k) = \oint_{\mathrlap{\{\hat \theta \in S^{2} : \hat \theta \cdot \hat k = 0\}}} \; \mathrm d^{2} \hat \theta \;\;\;\, d_{\pm}(\hat \theta) 
    = 4 \arcsin \left(\frac{\sin (\Omega_v/2)}{\max\left\lbrace
            \sin (\Omega_v/2) , \sin k_\theta
        \right\rbrace} \right) \, .
\end{equation}
The integral \eqref{eq:ex_cone_lineintegral} has closed form as follows.
Select a source point $x$ and volume point $r$. Define the angle $\theta_h \in [0, \pi]$ as $\theta_h = \angle(r_\rho, r_\psi, 0)(x_\rho, x_\psi,0)(0,0,0)$, illustrated in fig.~\ref{fig:ex_cone_cyl_b}.
Define $\theta_{xr} \in [0, \pi]$ to be the polar angle of $\hat \theta_{xr}$, illustrated in fig.~\ref{fig:ex_cone_cyl_c}.
Then \eqref{eq:ex_cone_lineintegral} can be evaluated in closed form to:
\begin{equation}
    \frac 1 2 \int_{-\infty}^{\infty} \mathrm d s \, s^2 \chi((r_\rho, r_\psi, r_z)+s\hat \theta_{xr}) = C R^2 (\csc \theta_{xr})^3 \frac{\cos(2 \theta_h) + (r_\rho/R)^2}{|\cos \theta_h|} \, .
\end{equation}

The backprojection-convolution formula \eqref{eq:th:cone_gbpf_inv_formula} may now be applied to recover the minimum-error solution to the attenuating volume described in \eqref{eq:th:cone_gbpf_err}.

A performant implementation of this reconstruction method has been developed and is described in \cite{grewar2024preprint}.

\section{Proofs}\label{sec:proofs}

The complete proof of theorems \ref{th:parallel_gbpf} and \ref{th:cone_gbpf} is given in this section. First there are some extra definitions required, given in \S\ref{sec:definitions_extra}. Then, there is a significant undertaking of proof of intermediate lemmas and corollaries in \S\ref{sec:proofs_bc}, culminating in a result pertaining to the backprojection-convolution algorithm. Then, there are some more intermediate results given in \S\ref{sec:proofs_cb} relating to the convolution-backprojection algorithm. Finally, the two major theorems from \S\ref{sec:theorems_01}, which comprise the main results of this article, are proven in \S\ref{sec:proofs_of_main_theorems} by drawing together the preceding intermediate results.

\subsection{Extra definitions}\label{sec:definitions_extra}

\begin{definition}\label{def:sdp}
    We define the single-direction projection operator $A_{\hat \theta} : \mathcal A \rightarrow \mathcal M_{\hat \theta}$ (where $\hat \theta \in S^{n-1}$) by
    \begin{align}
        A_{\hat \theta} = \Biggl\llbracket_{\mathcal A}^{\mathcal M_{\hat \theta}}
            \mu &\mapsto \Bigl\llbracket_{\mathbb R^n}^{\mathbb C}
                r_\perp \mapsto 
                \int_{-\infty}^\infty \mathrm d s \, \mu(r_\perp + s \hat \theta) 
            \Bigr\rrbracket \Biggr\rrbracket \, .
    \end{align}
\end{definition}
\begin{definition}[$A^{\dag_f}$]
    Let $f$ be a nonnegative, real distribution on $M$.
    If $f$ is positive, then define $A^{\dag_f}$ as the adjoint of $A$ with respect to the following inner product on $\mathcal M$:
    \begin{equation}
        \left\langle a, b \right\rangle_f = \oint_{\mathrlap{S^{n-1}}} \; \mathrm d^{n-1} \hat \theta \int_{\hat \theta^\perp} \mathrm d^{n-1} r_\perp a(\hat \theta, r_\perp)^* f(\hat \theta, r_\perp) b(\hat \theta, r_\perp) \, .
    \end{equation}
    If $f$ is not positive (i.e. it is $0$ on a set of nonzero measure), then $A^{\dag_f}$ is defined by a limiting process using positive distributions that approach $f$.
\end{definition}
\begin{definition}[$A_{\hat \theta}^{\dag_h}$]
    Let $h$ be a nonnegative, real distribution on $M_{\hat \theta}$.
    If $h$ is positive, then define $A^{\dag_h}_{\hat \theta}$ as the adjoint of $A_{\hat \theta}$ with respect to the following inner product on $\mathcal M_{\hat \theta}$:
    \begin{equation}
        \left\langle a, b \right\rangle_h = \int_{\hat \theta^\perp} \mathrm d^{n-1} r_\perp a(r_\perp)^* h(r_\perp) b(r_\perp) \, .
    \end{equation}
    If $h$ is not positive (i.e. it is $0$ on a set of nonzero measure), then $A_{\hat \theta}^{\dag_h}$ is defined by a limiting process using positive distributions that approach $h$.
    
    Furthermore, if $h$ is a distribution on $M$, i.e. $h \equiv h(\hat \theta, r)$, then define $A^{\dag_h}_{\hat \theta}$ in the natural manner with the first argument $h$ filled in with $\hat \theta$, i.e. $A^{\dag_h}_{\hat \theta} \equiv A^{\dag_{h(\hat \theta, \cdot)}}_{\hat \theta}$.
\end{definition}

\begin{definition}
    We use the symbol $\delta^n_a$ to denote the $n-$dimensional Dirac-delta distribution centred at $a \in \mathbb R^n$.
\end{definition}

\subsection{Proofs related to backprojection-convolution}\label{sec:proofs_bc}
\begin{lemma}\label{lem:sdp_adjoint}
The operator $A_{\hat \theta}^{\dag_h} : \mathcal M_{\hat \theta} \rightarrow \mathcal A$ is given explicitly by
    \begin{align}
            A_{\hat \theta}^{\dag_h} = \Biggl\llbracket_{\mathcal M_{\hat \theta}}^{\mathcal A}
            m &\mapsto \Bigl\llbracket_{\mathbb R^n}^{\mathbb C}
                r \mapsto m(r) h(r)
            \Bigr\rrbracket
            \Biggr\rrbracket \, .
    \end{align}
\end{lemma}
\begin{proof}
    $A_{\hat \theta}^{\dag_h}$ is defined by the following equality of inner products, where $m \in \mathcal M_{\hat \theta}$ and $\mu \in \mathcal A$:
    \begin{equation}
        \left\langle A_{\hat \theta}^{\dag_h} m , \mu \right\rangle
        = \left\langle m , A_{\hat \theta}\mu \right\rangle_h \, .
    \end{equation}
    Substituting $\mu$ with $\delta^n_r$ and using the $L^2$ norm of $\mathcal A$, the left-hand side becomes $(A_{\hat \theta}^{\dag_h} m)(r)$. Then, solving for the right-hand side we find:
    \begin{subequations}
    \begin{align}
        (A_{\hat \theta}^{\dag_h} m)(r)
        &= \left\langle m , A_{\hat \theta} \delta^n_{r} \right\rangle_h \\
        &= \left\langle m , \Bigl\llbracket_{M_{\hat \theta}}^{\mathbb C} r_\perp \mapsto 
            \int_{-\infty}^{\infty} \mathrm d s \, \delta_r^n(r_\perp + s \hat \theta)
            \Bigr\rrbracket
        \right\rangle_h \\
        \intertext{The line integral through an $n-$dimensional Dirac-delta distribution becomes an $(n-1)-$dimensional Dirac-delta distribution in the $\mathbb R^{n-1}$ space orthogonal to the direction of the line. The projection of $r$ into the $\mathbb R^{n-1}$ subspace is $(r - \hat \theta(\hat \theta \cdot r))$.}
        &= \left\langle m , \Bigl\llbracket_{M_{\hat \theta}}^{\mathbb C} r_\perp \mapsto 
            \delta_0^{n-1}\left(r_\perp - (r - \hat \theta(\hat \theta \cdot r))\right)
            \Bigr\rrbracket
        \right\rangle_h \\
        &= \int_{\mathrlap{\mathbb R^{n-1}}} \mathrm d^{n-1} r_\perp \; m(r_\perp) h(r_\perp) \delta_0^{n-1}\left(r_\perp - (r - \hat \theta(\hat \theta \cdot r))\right) \\
        &= h(r - \hat \theta(\hat \theta \cdot r)) m(r - \hat \theta(\hat \theta \cdot r)) \\
        &= h(r) m(r) \, .
    \end{align}
    \end{subequations}
\end{proof}

\begin{corollary}\label{cor:ata_single}
    \begin{align}
        A_{\hat \theta}^{\dag_h} A_{\hat \theta} =
        \Biggl\llbracket_{\mathcal A}^{\mathcal A}
        \mu &\mapsto \Bigl\llbracket_{\mathbb R^n}^{\mathbb C} r \mapsto 
        h(r) \int_{-\infty}^\infty \mathrm d s \, \mu(r + s \hat \theta) 
        \Bigr\rrbracket \Biggr\rrbracket
    \end{align}
\end{corollary}
\begin{proof}
    The definition for $A_{\hat \theta}$ (definition~\ref{def:sdp}) is directly applied to the expression for $A_{\hat \theta}^{\dag_h}$ given in lemma~\ref{lem:sdp_adjoint}, which shows that
    \begin{equation}
        A_{\hat \theta}^{\dag_h} A_{\hat \theta} = \Biggl\llbracket_{\mathcal A}^{\mathcal A} \mu \mapsto \Bigl\llbracket_{\mathbb R^n}^{\mathbb C} r \mapsto 
        h(r) \int_{-\infty}^\infty \mathrm d s \, \mu(r - \hat \theta (\hat \theta \cdot r) + s \hat \theta) 
        \Bigr\rrbracket \Biggr\rrbracket \, .
    \end{equation}
    The integration variable $s$ is then substituted with $s' = s - \hat \theta \cdot r$.
\end{proof}

\begin{lemma}\label{lem:full_back_proj}
    \begin{align}
        A^{\dag_f} = \Biggl\llbracket_{\mathcal M}^{\mathcal A}
        m &\mapsto \Bigl\llbracket_{\mathbb R^n}^{\mathbb C}
                r \mapsto \oint_{\mathrlap{S^{n-1}}} \; \mathrm d^{n-1} \hat \theta f(\hat \theta, r) m(\hat \theta, r)
            \Bigr\rrbracket \Biggr\rrbracket \, .
    \end{align}
\end{lemma}
\begin{proof}
    The proof is nearly identical to that for lemma \ref{lem:sdp_adjoint}, except that the inner product $f$ on $\mathcal M$ is used, which introduces the integral and the factor of $f(...)$ in the integrand. We will repeat the proof now with the necessary alterations.

    The adjoint of $A^{\dag_f}$ is defined by the following equality of inner products, where $m \in \mathcal M$ and $\mu \in \mathcal A$:
    \begin{equation}
        \left\langle A^{\dag_f} m , \mu \right\rangle
        = \left\langle m , A \mu \right\rangle_f \, .
    \end{equation}
    Substituting $\mu$ with $\delta^n_r$ and using the $L^2$ norm of $\mathcal A$, the left-hand side becomes $(A^{\dag_f} m)(r)$. Then, solving for the right-hand side we find:
    \begin{subequations}
    \begin{align}
        (A^{\dag_f} m)(r)
        &= \left\langle m , A \delta^n_{r} \right\rangle_f \\
        &= \left\langle m , \Biggl\llbracket_{M}^{\mathbb C} (\hat \theta, r_\perp) \mapsto 
            \int_{-\infty}^{\infty} \mathrm d s \, \delta_r^n(r_\perp + s \hat \theta) \Biggr\rrbracket
        \right\rangle_f \\
        \intertext{The line integral through an $n-$dimensional Dirac-delta distribution becomes an $(n-1)-$dimensional Dirac-delta distribution in the $\mathbb R^{n-1}$ space orthogonal to the direction of the line. The projection of $r$ into the $\mathbb R^{n-1}$ subspace is $(r - \hat \theta(\hat \theta \cdot r))$.}
        &= \left\langle m , \Biggl\llbracket_{M}^{\mathbb C} (\hat \theta, r_\perp) \mapsto 
            \delta_0^{n-1}\left(r_\perp - (r - \hat \theta(\hat \theta \cdot r))\right) \Biggr\rrbracket
        \right\rangle_f \\
        &= \oint_{S^{n-1}} \mathrm d^{n-1} \hat \theta \int_{\mathrlap{\mathbb R^{n-1}}} \mathrm d^{n-1} r_\perp \; m(\hat \theta, r_\perp) f(\hat \theta, r_\perp) \delta_0^{n-1}\left(r_\perp - (r - \hat \theta(\hat \theta \cdot r))\right) \\
        &= \oint_{S^{n-1}} \mathrm d^{n-1} \hat \theta f(\hat \theta, r - \hat \theta(\hat \theta \cdot r)) m(\hat \theta, r - \hat \theta(\hat \theta \cdot r)) \\
        &= \oint_{S^{n-1}} \mathrm d^{n-1} \hat \theta f(\hat \theta, r)m(\hat \theta, r) \, .
    \end{align}
    \end{subequations}
\end{proof}

\begin{lemma}\label{lem:backproj}
    Let $m \in \mathcal M$. 
    Let $f$ be a real, nonnegative distribution on $M$.
    Define $m_{\hat \theta}(r) = m(\hat \theta, r)$. Likewise, define $f_{\hat \theta}(r) = f(\hat \theta, r)$.
    Then the following identity holds:
    \begin{equation}
        A^{\dag_f} m = \oint_{\mathrlap{S^{n-1}}} \; \, \mathrm d^{n-1} \hat \theta \, A_{\hat \theta}^{\dag_{f_{\hat \theta}}} m_{\hat \theta} \, .
    \end{equation}
    \begin{proof}
        The expression for $(A_{\hat \theta}^{\dag_{f_{\hat \theta}}}m_{\hat \theta})(r)$ is computed from lemma~\ref{lem:sdp_adjoint} and is equal to $f_{\hat \theta}(r) m_{\hat \theta}(r) = f(\hat \theta, r) m(\hat \theta, r)$. It is seen from lemma~\ref{lem:full_back_proj} that 
        \begin{equation}
        (A^{\dag_f}m)(r) = \oint \mathrm d^{n-1} \hat \theta f(\hat \theta, r) m(\hat \theta, r) \, .
        \end{equation}
        The proof is then immediate.
    \end{proof}
\end{lemma}

\begin{lemma}\label{lem:ATA_1}
    Let $f$ be a real, nonnegative distribution on $M$. 
    Define $f_{\hat \theta}(r) = f(\hat \theta, r)$.
    Then
    \begin{equation}
        A^{\dag_f} A = \oint_{\mathrlap{S^{n-1}}} \; \mathrm d^{n-1} \hat \theta \, A_{\hat \theta}^{\dag_{f_{\hat \theta}}} A_{\hat \theta} \, .
    \end{equation}
\end{lemma}
\begin{proof}
    The left-hand side may be evaluated by a direct application of definition~\ref{def:xray_transform} and lemma~\ref{lem:full_back_proj}. The right-hand side may be evaluated by a direct application of corollary~\ref{cor:ata_single}. Each side evaluates to the following operator on $\mathcal A$:
    \begin{equation}
    \Biggl\llbracket_{\mathcal A}^{\mathcal A} \mu \mapsto \Bigl\llbracket_{\mathbb R^n}^{\mathbb C} r \mapsto
            \oint_{\mathrlap{S^{n-1}}} \; \mathrm d^{n-1} \hat \theta \, f(\hat \theta, r) \int_{-\infty}^\infty \mathrm d s \mu(r + s \hat \theta)
    \Bigr\rrbracket \Biggr\rrbracket \, .
    \end{equation}
\end{proof}

\begin{lemma}[A variant of the Fourier Slice Theorem]\label{lem:fst}
    Let $h$ be a nonnegative, real distribution on $M_{\hat \theta}$. Then the following identity holds:
    \begin{equation}
        A_{\hat \theta}^{\dag_h} A_{\hat \theta} = F_{\mathcal A}^{-1} \mathrm{diag}_{k}\left( 2 \pi  \delta(k \cdot \hat \theta) \right) F_{\mathcal A} \mathrm{diag}_{r}\left(
                h(r)
            \right) \, .
    \end{equation}
    \begin{proof}
        Using the explicit expression for $A_{\hat \theta}^{\dag_h} A_{\hat \theta}$ from corollary~\ref{cor:ata_single}, we first write the operator:
        \begin{equation}
            F_{\mathcal A} A_{\hat \theta}^{\dag_h} A_{\hat \theta} = \Biggl\llbracket_{\mathcal A}^{\mathcal A} \mu \mapsto \Bigl\llbracket_{\mathbb R^n}^{\mathbb C}
                k \mapsto (\sqrt{2 \pi})^{-n} \int_{\mathbb R^n} \mathrm d^n r \,
                    e^{-i r \cdot k} h(r) \int_{-\infty}^\infty \mathrm d s \, \mu(r + s \hat \theta)
            \Bigr\rrbracket \Biggr\rrbracket \, .
        \end{equation}
        We then reorder the integrations, translate the integration variable $r$, then reorder the integrations again. We will justify in a moment why it is valid to reorder the integrations in this manner.
        \begin{align*}
            & \int_{\mathbb R^n} \mathrm d^n r \,
                    e^{-i r \cdot k} h(r) \int_{-\infty}^\infty \mathrm d s \, \mu(r + s \hat \theta) \\
            &= \int_{-\infty}^\infty \mathrm d s \int_{\mathbb R^n} \mathrm d^n r \, e^{-i r \cdot k} h(r) \mu(r + s \hat \theta) \\
            &= \int_{-\infty}^\infty \mathrm d s \int_{\mathbb R^n} \mathrm d^n r \, e^{-i (r - s \hat \theta) \cdot k} h(r - s \hat \theta) \mu(r) \\
            &= \int_{-\infty}^\infty \mathrm d s \int_{\mathbb R^n} \mathrm d^n r \, e^{-i (r - s \hat \theta) \cdot k} h(r) \mu(r) \\
            &= \int_{-\infty}^\infty \mathrm d s \int_{\mathbb R^n} \mathrm d^n r \, e^{-i r \cdot k} e^{i s \hat \theta \cdot k} h(r) \mu(r) \, .
        \end{align*}
        We then use the well-known identity equating the following integral over $s$ with a Dirac-delta distribution when it is contained inside another integration:\footnote{
            The integral in this identity does not converge, but when it is part of a higher-dimensional integration which is known to converge, then it is valid to use this identity.}
        \begin{equation}
            \frac{1}{2 \pi} \int_{-\infty}^\infty \mathrm d s \, e^{isa} = \delta(a) \, .
        \end{equation}
        Formally, the validity of these manipulations may be justified by considering a limiting sequence of functions that approach the Dirac-delta distribution, and applying the Lebesgue dominated convergence theorem. (A similar integration is performed in a similar context with similar reasoning in \cite{katsevich2002theoretically}.)
        Applying the above identity, we find
        \begin{equation}
            F_{\mathcal A} A_{\hat \theta}^{\dag_h} A_{\hat \theta} = \Biggl\llbracket_{\mathcal A}^{\mathcal A} \mu \mapsto \Bigl\llbracket_{\mathbb R^n}^{\mathbb C}
                k \mapsto 
                    2 \pi \delta(k \cdot \hat \theta) \int_{\mathbb R^n} \mathrm d^n r \, h(r) \mu(r) e^{-i r \cdot k}
                \Bigr\rrbracket \Biggr\rrbracket  \, .
        \end{equation}
        Comparing this operator with the definition for $F_{\mathcal A}^{-1}$, the following conclusion is immediate:
        \begin{equation}
            F_{\mathcal A} A_{\hat \theta}^{\dag_h} A_{\hat \theta} = \mathrm{diag}_{k}(2 \pi \delta(k \cdot \hat \theta)) 
            F_{\mathcal A} 
            \mathrm{diag}_{r}\left(
                h(r)
            \right)\, .
        \end{equation}
        We then apply $F_{\mathcal A}^{-1}$ to the left of both sides, and use $F_{\mathcal A}^{-1} F_{\mathcal A} = 1$ to recover the identity
        \begin{equation}
            A_{\hat \theta}^{\dag_h} A_{\hat \theta} = F_{\mathcal A}^{-1} \mathrm{diag}_{k}(2 \pi \delta(k \cdot \hat \theta)) 
            F_{\mathcal A} 
            \mathrm{diag}_{r}\left(
                h(r)
            \right)\, .
        \end{equation}
    \end{proof}
\end{lemma}

\begin{corollary}[{Decomposition of $A^{\dag_f} A$}]\label{cor:ATA_decomp}
    For all distributions $f$, the forward-backprojection $A^{\dag_f} A$ can be decomposed as follows:
    \begin{equation}
        A^{\dag_f} A = \oint_{\mathrlap{S^{n-1}}} \; \mathrm d^{n-1} \hat \theta \,
        F_{\mathcal A}^{-1} \mathrm{diag}_{k}(2 \pi \delta(k \cdot \hat \theta)) F_{\mathcal A} 
        \mathrm{diag}_{r}\left(
            f(\hat \theta, r)
        \right) \, .
    \end{equation}
    \begin{proof}
        The equation is obtained by substituting lemma~\ref{lem:ATA_1} into lemma~\ref{lem:fst}.
    \end{proof}
\end{corollary}

\begin{lemma}\label{lem:ATA_decomp_simplified}
    For distributions $f(\hat \theta, r)$ of the form $f(\hat \theta, r) = d(\hat \theta)$, the following identity holds:
    \begin{equation}
        A^{\dag_f} A = F_{\mathcal A}^{-1} \mathrm{diag}_k\left(
            2 \pi |k|^{-1} \oint_{\mathrlap{\{\hat \theta \in S^{n-1} : \hat \theta \cdot k = 0\}}} \; \mathrm d^{n-2} \hat \theta \;\;\;\, d(\hat \theta)
            \right) F_{\mathcal A} \, . 
    \end{equation}
    \begin{proof}
        We begin with corollary~\ref{cor:ATA_decomp}.
        In the circumstance where $f(\hat \theta, r) = d(\hat \theta)$, it is possible to pull the integration over $\hat \theta$ into the diagonal operator acting on the Fourier basis viz.
        \begin{subequations}
        \begin{align}
            A^{\dag_f} A &= \oint_{\mathrlap{S^{n-1}}} \; \mathrm d^{n-1} \hat \theta \,
            F_{\mathcal A}^{-1} \mathrm{diag}_{k}(2 \pi \delta(k \cdot \hat \theta)) F_{\mathcal A} 
            \mathrm{diag}_{r}\left(
                f(\hat \theta, r)
            \right) \\
            &= \oint_{\mathrlap{S^{n-1}}} \; \mathrm d^{n-1} \hat \theta \,
            F_{\mathcal A}^{-1} \mathrm{diag}_{k}(2 \pi \delta(k \cdot \hat \theta)) F_{\mathcal A} 
            \mathrm{diag}_{r}\left(
                d(\hat \theta)
            \right) \\
            &= \oint_{\mathrlap{S^{n-1}}} \; \mathrm d^{n-1} \hat \theta \,
            F_{\mathcal A}^{-1} \mathrm{diag}_{k}(2 \pi \delta(k \cdot \hat \theta)) F_{\mathcal A} 
            d(\hat \theta) \\
            &= \oint_{\mathrlap{S^{n-1}}} \; \mathrm d^{n-1} \hat \theta \,
            F_{\mathcal A}^{-1} \mathrm{diag}_{k}\left(2 \pi \delta(k \cdot \hat \theta) d(\hat \theta)\right) F_{\mathcal A} \\
            &= F_{\mathcal A}^{-1} \mathrm{diag}_{k}\left(2 \pi \oint_{\mathrlap{S^{n-1}}} \; \mathrm d^{n-1} \hat \theta \, \delta(k \cdot \hat \theta) d(\hat \theta)\right) F_{\mathcal A} \label{eq:ATA_proof_01} \, .
        \end{align}
        \end{subequations}
        In this form, it is possible to analytically integrate out the Dirac-delta distribution. To do this, we introduce the hyperspherical coordinate system for $S^{n-1}$ found in \cite{blumenson1960derivation}. A point $\hat \theta \in S^{n-1}$ is given by $\hat \theta = (\phi_1, \phi_2, \dots \phi_{n-1})$. We arrange the coordinate system so that the `top' of the sphere, defined by $\phi_1 = 0$, aligns with $k \in \mathbb R^n$. The Dirac-delta distribution $\delta(k \cdot \hat \theta)$ becomes
        \begin{equation}
            \delta(k \cdot \hat \theta) = \delta(|k| \cos \phi_1) \, .
        \end{equation}
        The coordinates $\phi_i$ are locally orthogonal everywhere, and for small variations in the $\phi_1$ coordinate alone, the distance swept out on the sphere is $\mathrm d s^2 = \mathrm d \phi_1^2$. In other words, the metric matrix $g$ of the hyperspherical coordinate system is diagonal, with the top-left entry $g_{11}$, corresponding to $\phi_1$, equal to $1$. Therefore, an integration over $S^{n-1}$ of the following form may be simplified as follows, where $x(\hat \theta)$ is arbitrary:
        \begin{subequations}
        \begin{align}
            \oint_{\mathrlap{S^{n-1}}} \; \mathrm d^{n-1} \hat \theta \, x(\hat \theta) \delta(k \cdot \hat \theta) 
            &= \oint_{{S^{n-1}}} \left(\prod_{i=1}^{n-1} g_{ii} \mathrm d \phi_i\right) \, x(\phi_1, \dots, \phi_{n-1}) \delta(|k| \cos \phi_1) \\
            &= \int_0^\pi \mathrm d \phi_1 \oint_{{S^{n-2}}} \left(\prod_{i=2}^{n-1} g_{ii} \mathrm d \phi_i\right) \, x(\phi_1, \dots, \phi_{n-1}) \delta(|k| \cos \phi_1) \, .
        \end{align}
        \end{subequations}
        The value of $\phi_1$ is everywhere the same on the support of the Dirac-delta distribution: $\phi_1 = \pi/2$. For this reason, it is valid to replace $\phi_1$ with $\pi/2$ everywhere inside the integrand, except inside the Dirac-delta distribution. Then, we may reorder the integration and evaluate like so
        \begin{subequations}
        \begin{align}
            & \quad\int_0^\pi \mathrm d \phi_1 \oint_{{S^{n-2}}} \left(\prod_{i=2}^{n-1} g_{ii} \mathrm d \phi_i\right) \, x(\pi/2, \dots, \phi_{n-1}) \delta(|k| \cos \phi_1) \\
            &= \oint_{{S^{n-2}}} \left(\prod_{i=2}^{n-1} g_{ii} \mathrm d \phi_i\right) \, x(\pi/2, \phi_2, \dots, \phi_{n-1}) \int_0^\pi \mathrm d \phi_1  \delta(|k| \cos \phi_1) \\
            &= \oint_{{S^{n-2}}} \left(\prod_{i=2}^{n-1} g_{ii} \mathrm d \phi_i\right) \, x(\pi/2, \phi_2, \dots, \phi_{n-1}) |k|^{-1} \, .
        \end{align}
        \end{subequations}
        The remaining integration over the coordinates $(\phi_2, \phi_3, \dots \phi_{n-1})$ is an integration over a subsphere $S^{n-2} \subset S^{n-1}$ with the inherited metric. Therefore, we have the identity 
        \begin{align*}
            \oint_{\mathrlap{S^{n-1}}} \; \mathrm d^{n-1} \hat \theta \,
            x(\hat \theta) \delta(k \cdot \hat \theta) 
            &= |k|^{-1} \oint_{\mathrlap{\{\hat \theta \in S^{n-1} : \hat \theta \cdot k = 0\}}} \; \mathrm d^{n-2} \hat \theta \;\;\;\,
            x(\hat \theta) \, .
        \end{align*}
        Applying this identity to \eqref{eq:ATA_proof_01} completes the proof.\footnote{
            One subtle point is how to reconcile $|k|^{-1}$ for $k = 0$. However, functions in $\mathcal A$ must have finite norm and so the zero-frequency component of the Fourier transform is always $0$. The Fourier transform $F_{\mathcal A}$ may be interpreted as lacking a $k = 0$ component, or one may simply treat $|k|^{-1}$ as if it were $0$ when $k = 0$.}
    \end{proof}
\end{lemma}

\begin{lemma}\label{lem:inv_ATA}
Let $f(\hat \theta, r_\perp)$ be a distribution of the form $f(\hat \theta, r_\perp) = d(\hat \theta)$. The forward-backprojection $A^{\dag_f} A$ has an inverse if and only if
    \begin{equation}
        \forall k \in (\mathbb R^n - \{0\}), \qquad 0 \neq \oint_{\mathrlap{\{\hat \theta \in S^{n-1} : \hat \theta \cdot k = 0\}}} \; \mathrm d^{n-2} \hat \theta \;\;\;\, d(\hat \theta) \, .
    \end{equation}
    (In words: the integration of $d(\hat \theta)$ around any given great-circle must be nonzero.)
    When the inverse exists, it is given by
    \begin{equation}
        (A^{\dag_f} A)^{-1} = F_{\mathcal A}^{-1} \mathrm{diag}_k\left(
            \frac{|k|}{2 \pi} \left/ \oint_{\mathrlap{\{\hat \theta \in S^{n-1} : \hat \theta \cdot k = 0\}}} \; \mathrm d^{n-2} \hat \theta \;\;\;\, d(\hat \theta) \right.
            \right) F_{\mathcal A} \, . 
    \end{equation}
    \begin{proof}
        The proof is straightforward from lemma~\ref{lem:ATA_decomp_simplified}. 
    \end{proof}
\end{lemma}

\begin{lemma}\label{lem:min_norm}
    Let $\mu \in \mathcal A$ and $m \in \mathcal M$ be real-valued functions. Let $f$ be a real, nonnegative distribution on $M$.
    If and only if the inverse $(A^{\dag_f} A)^{-1}$ exists, then the following `error function' between $m$ and $A \mu$ has a unique minimum solution in $\mathcal A$, given on the right-hand side:
    \begin{equation}\label{eq:mpinv_th}
        \mathrm{arg~min}_{\mu \in \mathcal A} \oint_{\mathrlap{S^{n-1}}} \;\mathrm d^{n-1} \hat \theta \int_{\mathrlap{\hat \theta^\perp}} \mathrm d^{n-1} r_\perp f(\hat \theta, r_\perp) \left( (m-A\mu)(\hat \theta, r_\perp) \right)^2
        =
        (A^{\dag_f} A)^{-1} A^{\dag_f} m \, .
    \end{equation}
    \begin{proof}
    We prove this for strictly positive $f$, from which the answer follows for nonnegative $f$ by a limiting process.
    Assuming $f$ is positive, we may identify the following integral expression as an inner product between vectors in $\mathcal M$:
    \begin{equation}
        \oint_{\mathrlap{S^{n-1}}} \;\mathrm d^{n-1} \hat \theta \int_{\mathrlap{\hat \theta^\perp}} \mathrm d^{n-1} r_\perp f(\hat \theta, r_\perp) a^*(\hat \theta, r_\perp) b(\hat \theta, r_\perp)
        =
        \langle a, b \rangle_f \, .
    \end{equation}
    Equation \eqref{eq:mpinv_th} may thus be rewritten:
    \begin{equation}
        \mathrm{arg~min}_{\mu \in \mathcal A} \langle m - A \mu , m - A \mu \rangle_f
        =
        (A^{\dag_f} A)^{-1} A^{\dag_f} m \, .
    \end{equation}
    We now use the calculus of variations to prove this equation correct.\footnote{For those unfamiliar: a `variation' is an infinitesimal change. When applying variations to differentiable functions, the change in the function value is equal to the change predicted by its derivative. E.g. for a differentiable function $f : \mathbb R \rightarrow \mathbb R$, we have $\delta f(a) = f'(a) \delta a$. When variational equations are used, there is some implicit definition about what is being varied, e.g. in the previous equation, we were not allowing $f$ itself to vary. If we did allow $f$ to vary, then we'd have $\delta f(a) = f'(a) \delta a + (\delta f) (a)$, because the product rule of differentiation applies, and second-order derivatives are effectively zero because the variation is taken to be an infinitesimally small change.} 
    We use the following basic identities related to variation: 
    \begin{subequations}
    \begin{align}
        \delta\langle a, b \rangle &= \langle \delta a , b\rangle + \langle a , \delta b \rangle \, , \\
        \delta (ab) &= (\delta a) b + a (\delta b) \, .
    \end{align}
    \end{subequations}
    The choice of $\mu$ which minimises $\langle m - A \mu , m - A \mu \rangle_f$ is determined by solving for the $\mu$ such that the variation is $0$ under all possible variations $\delta \mu$. Expanding the variation explicitly, we have:
    \begin{align}
        \delta \left\langle 
                m - A \mu, m - A \mu
        \right\rangle_f
        &= \delta \left( \Big.
            \langle m, m \rangle_f + \left\langle 
                A \mu, A \mu
            \right\rangle_f - 2 \left\langle m, A \mu \right\rangle_f
        \right) \\
        &= \delta \langle m, m \rangle_f + \delta \left\langle 
                A \mu, A \mu
            \right\rangle_f - 2 \delta \left\langle m, A \mu \right\rangle_f \\
        &= 2 \left\langle 
                A \mu, A \delta \mu
            \right\rangle_f - 2 \left\langle m, A \delta\mu \right\rangle_f \\
        &= 2 \left\langle 
                A \mu - m, A \delta \mu
            \right\rangle_f \\
        &= 2 \left\langle 
                A^{\dag_f}(A \mu - m), \delta \mu
            \right\rangle_f \, . \label{eq:vary_01}
    \end{align}
    In the last line, we have used the fact that $A^{\dag_f}$ is the adjoint of $A$ with respect to the inner product $\langle \cdot, \cdot \rangle_f$ between measurement vectors.
    Since $\delta \mu$ is free to vary without any constraint, \eqref{eq:vary_01} can only be zero for all variations $\delta \mu$ when the left term of the inner product is zero. We therefore find that a value of $\mu$ satisfying
    \begin{equation}
        A^{\dag_f}A \mu = A^{\dag_f} m
    \end{equation}
    achieves the minimum value for the error function.
    Assuming that the inverse $(A^{\dag_f}A)^{-1}$ exists, we can then conclude that
    \begin{equation}
        \mu = (A^{\dag_f}A)^{-1} A^{\dag_f} m \, .
    \end{equation}
    If the inverse $(A^{\dag_f}A)^{-1}$ does not exist, then there are multiple solutions for $\mu$ which differ by elements of the nullspace of $A^{\dag f} A$.
    \end{proof}
\end{lemma}

\subsection{Proofs related to convolution-backprojection}\label{sec:proofs_cb}

\begin{identity}[A variant of the Fourier Slice Theorem]\label{id:vol_proj_filt}
    Let $h$ be a real, nonegative distribution on $M_{\hat \theta}$.
    Then for any function $g : \mathbb R^n \rightarrow \mathbb C$,
    \begin{equation}
        F_{\mathcal A}^{-1} \mathrm{diag}_k(g(k)) F_{\mathcal A} A_{\hat \theta}^{\dag_h} \!=\! \mathrm{diag}_{r_\perp} \!\! \left(h(r_\perp)\right)^{-1} A_{\hat \theta}^{\dag_h} F_{\hat \theta}^{-1} \mathrm{diag}_{k_\perp} \!\! \left(g(k_\perp)\right) F_{\hat \theta} \mathrm{diag}_{r_\perp} \!\! \left(h(r_\perp)\right) .
    \end{equation}
    \begin{proof}
        We first write out the following operator using lemma~\ref{lem:sdp_adjoint} to describe $A_{\hat \theta}^{\dag_h}$:
        \begin{align}
            &(F_{\mathcal A} A_{\hat \theta}^{\dag_h})(p) 
            = \Biggl\llbracket_{\mathbb R^n}^{\mathbb C} k \mapsto           
                (\sqrt{2\pi})^{-n} \int_{\mathrlap{\mathbb R^n}} \mathrm d^{n} r \,
                h(r) p(r)
                e^{-ir \cdot k}
            \Biggr\rrbracket \, .
        \end{align}
        We split the integration into the one dimensional integration orthogonal to $\hat \theta$, and the rest of the dimensions. The point $r$ decomposes into $r_\perp + \hat \theta(\hat \theta \cdot r)$. Likewise, we decompose $k$ per $k = k_\perp + \hat \theta(\hat \theta \cdot k)$.
        \begin{subequations}
        \begin{align}
            \int_{\mathrlap{\mathbb R^n}} \mathrm d^{n} r \,
                h(r) p(r)
                e^{-ir \cdot k} 
            &= \int_{\mathrlap{\hat \theta^\perp}} \mathrm d^{n-1} r_\perp \!\!
                \int_{\mathrlap{-\infty}}^{\mathrlap{\infty}} \!\! \mathrm d s \,
                h(r_\perp) p(r_\perp)
                e^{-i(r_\perp \! + s \hat \theta) \cdot k} \\
            &= \int_{\mathrlap{\hat \theta^\perp}} \mathrm d^{n-1} r_\perp \!\!
                \int_{\mathrlap{-\infty}}^{\mathrlap{\infty}} \!\! \mathrm d s \,
                h(r_\perp) p(r_\perp)
                e^{-i r_\perp \cdot k_\perp} e^{-i s \hat \theta \cdot k} \\
            \shortintertext{then, using $\int_{-\infty}^{\infty} \mathrm d s e^{-isa} = 2 \pi \delta(a)$,}
            &= 2 \pi \delta(\hat \theta \cdot k) \int_{\mathrlap{\hat \theta^\perp}} \mathrm d^{n-1} r_\perp \! \, 
                h(r_\perp) p(r_\perp)
                e^{-i r_\perp \cdot k_\perp}  \label{eq:fst_intermediate_01} \, .
        \end{align}
        \end{subequations}
        
        We operate on the result with $F_{\mathcal A}^{-1} \mathrm{diag}_k(g(k))$ to produce:
        \begin{subequations}
        \begin{align}
            &(F_{\mathcal A}^{-1} \mathrm{diag}_k(g(k)) F_{\mathcal A} A_{\hat \theta}^{\dag_h})(p)(r)\\
            &= \frac{1}{(2\pi)^n}  \int_{\mathrlap{\mathbb R^n}} \mathrm d^n k' \, g(k') 2 \pi \delta(\hat \theta \cdot k) e^{ik' \cdot r} \int_{\mathrlap{\hat \theta^\perp}} \mathrm d^{n-1} r_\perp' \! \, 
                h(r_\perp') p(r_\perp')
                e^{-i r_\perp' \cdot k_\perp'} \\
            &= \frac{1}{(2\pi)^{n-1}}  \int_{\mathrlap{\hat \theta^\perp}} \mathrm d^n k' \, g(k_\perp') e^{ik_\perp' \cdot r_\perp} \int_{\mathrlap{\hat \theta^\perp}} \mathrm d^{n-1} r_\perp' \! \, 
                h(r_\perp') p(r_\perp')
                e^{-i r_\perp' \cdot k_\perp'} \\
            &= \frac{1}{(2\pi)^{n-1}}  \int_{\mathrlap{\hat \theta^\perp}} \mathrm d^n k' \, g(k_\perp') \int_{\mathrlap{\hat \theta^\perp}} \mathrm d^{n-1} r_\perp' \! \, 
                h(r_\perp') p(r_\perp')
                e^{i (r_\perp - r_\perp') \cdot k_\perp'} \\
            &= \left(
            \mathrm{diag}_{r_\perp} \!\! \left(h(r_\perp)\right)^{-1} A_{\hat \theta}^{\dag_h} F_{\hat \theta}^{-1} \mathrm{diag}_{k_\perp} \!\! \left(g(k_\perp)\right) F_{\hat \theta} \mathrm{diag}_{r_\perp} \!\! \left(h(r_\perp)\right)
            \right)(p)(r) \, .
        \end{align}
        \end{subequations}
    \end{proof}
\end{identity}

\begin{lemma}[Filtered Backprojection form of Moore-Penrose Inverse]\label{lem:inv_ATA_fbp}
    Let $f(\hat \theta, r_\perp)$ be a real, nonnegative distribution of the form $f(\hat \theta, r_\perp) = d(\hat \theta)$. Define $f_{\hat \theta}(r_\perp) = f(\hat \theta, r_\perp) = d(\hat \theta)$. 
    If the forward-backprojection $A^{\dag_f} A$ has an inverse, then
    \begin{equation}
        (A^{\dag_f} A)^{-1} A^{\dag_f} = A^{\dag_f} F_{\mathcal M}^{-1} \mathrm{diag}_k\!\left(
            \frac{|k|}{2 \pi \mathcal I[d](k/|k|)}
            \right) F_{\mathcal M} \, .
    \end{equation}
    \begin{proof}
    First, using lemma~\ref{lem:inv_ATA}, we have
    \begin{subequations}
    \begin{align}
         &(A^{\dag_f} A)^{-1} A^{\dag_f} m = 
         F_{\mathcal A}^{-1} \mathrm{diag}_k\left(
            \frac{|k|}{2 \pi \mathcal I[d](k/|k|)}
            \right) F_{\mathcal A}
         A^{\dag_f} m \\
         \shortintertext{then we apply lemma~\ref{lem:backproj},}
         &= F_{\mathcal A}^{-1} \mathrm{diag}_k\left(
            \frac{|k|}{2 \pi \mathcal I[d](k/|k|)}
            \right) F_{\mathcal A}
         \oint_{\mathrlap{S^{n-1}}} \; \, \mathrm d^{n-1} \hat \theta \, A_{\hat \theta}^{\dag_{f_{\hat \theta}}} m(\hat \theta, \cdot) \\
         &= \oint_{\mathrlap{S^{n-1}}} \; \, \mathrm d^{n-1} \hat \theta \, F_{\mathcal A}^{-1} \mathrm{diag}_k\left(
            \frac{|k|}{2 \pi \mathcal I[d](k/|k|)}
            \right) F_{\mathcal A}
         A_{\hat \theta}^{\dag_{f_{\hat \theta}}} m(\hat \theta, \cdot) \\
         \shortintertext{then, using lemma~\ref{id:vol_proj_filt},}
         &= \oint_{\mathrlap{S^{n-1}}} \; \, \mathrm d^{n-1} \hat \theta \,
         A_{\hat \theta}^{\dag_{f_{\hat \theta}}} F_{\hat \theta}^{-1} \mathrm{diag}_{k_\perp} \!\!\! \left(
            \frac{|k_\perp|}{2 \pi \mathcal I[d](k_\perp/|k_\perp|)}
            \right) F_{\hat \theta} m(\hat \theta, \cdot) \\
        \shortintertext{then, using definition~\ref{def:full_proj_ft} for $F_{\mathcal M}$,}
        &= \oint_{\mathrlap{S^{n-1}}} \; \, \mathrm d^{n-1} \hat \theta \,
         A_{\hat \theta}^{\dag_{f_{\hat \theta}}} \left(F_{\mathcal M}^{-1} \mathrm{diag}_{k} \!\! \left(
            \frac{|k|}{2 \pi \mathcal I[d](k/|k|)}
            \right) F_{\mathcal M} m \right)(\hat \theta, \cdot) \\
        \shortintertext{then, using lemma~\ref{lem:backproj} once again,}
        &= A^{\dag_f} F_{\mathcal M}^{-1} \mathrm{diag}_{k} \!\! \left(
            \frac{|k|}{2 \pi \mathcal I[d](k/|k|)}
            \right) F_{\mathcal M} m \, .
    \end{align}
    \end{subequations}
    \end{proof}
\end{lemma}

\subsection{Proofs of the main theorems}\label{sec:proofs_of_main_theorems}

\begin{proof}[Proof of theorem~\ref{th:parallel_gbpf}]
    Define the distribution $f$ on $M$ given by $f(\hat \theta, r_\perp) = d(\hat \theta)$.
    Lemma~\ref{lem:inv_ATA} proves that $d(\hat \theta)$ satisfies eq.~\eqref{eq:th:gbpf_existence} precisely when the inverse operator $(A^{\dag_f}A)^{-1}$ exists.
    Lemma~\ref{lem:min_norm} proves that the expression for $\mu$ in eq.~\eqref{eq:th:gbpf_inv_formula} is the unique solution to eq.~\eqref{eq:th:gbpf_err} when the inverse operator $(A^{\dag_f}A)^{-1}$ exists, and that if the inverse operator does not exist, then the solution is not unique.
    Lemma~\ref{lem:inv_ATA_fbp} proves the correctness of the second formulation eq.~\eqref{eq:th:fbp_inv_formula}.
\end{proof}

\begin{proof}[Proof of theorem~\ref{th:cone_gbpf}]
    The proof begins from theorem~\ref{th:parallel_gbpf}, as this is a special case.

    We assume a density of X-ray source points given by a distribution $\chi$ on $\mathbb R^n$, and seek to change the integration measure of the weighted backprojection to $\int_{\mathbb R^n} \mathrm d x \, \chi(x)$.
    For a given point in the volume, $r \in \mathbb R^n$, the determinant of transformation from the $\hat \theta$ parameterisation of lines to the $x$ parameterisation is given by the integral of $\chi$ over the $n$-dimensional volume swept out per unit $(n-1)$-dimensional angle:\footnote{For $n=2$, it is a 2-dimensional volume swept out per radian. For $n=3$, it is a 3-dimensional volume swept out per steradian.}
    \begin{equation}
        {\mathrm d^n x \, \chi(x)}
        \equiv 
        \mathrm d^{n-1} \hat \theta 
            \int_{0}^\infty \mathrm d s \; s^{n-1} \chi(r + s \hat \theta) \, .
    \end{equation}
    Therefore, the weighted backprojection of $m \in \mathcal M$ becomes
    \begin{equation}
                \Biggl\llbracket_{\mathbb R^n}^{\mathbb C} r \mapsto \int_{\mathrlap{\mathbb R^n}} \; \mathrm d^n x \, \chi(x) \;
                \left( \int_{0}^\infty \mathrm d s \; s^{n-1} \chi(r + s \hat \theta) \right)^{-1} \;
                d(\hat \theta) m(\hat \theta, x) \Biggr\rrbracket
                \, ,
    \end{equation}
    where $\hat \theta$ denotes the direction of the vector $(r-x)$.
    However, there is an issue with the practical application of this backprojection formula. The issue is that the inner integral can be $0$---particularly when $r=x$---causing the backprojection to diverge. 
    This is a serious impediment to computationally implementing the backprojection-convolution formula, because the backprojected data is represented as a function on a discrete lattice which is unable to adequately represent such singularities. The subsequent convolution cannot then be numerically evaluated. To be pedantic: this problem may be averted by pathological choices of source densities $\chi(x)$, but it will not be avoided for any reasonable choice.\footnote{
        Let $S \subseteq \mathbb R^n$ be the closure of the chord hull of the support of $\chi$. If the boundary $\partial S$ is non-empty (as it will be for any realistic choice of $\chi$), then the boundary points $r \in \partial S$ will suffer from this issue wherever (if ever) $d(\hat \theta) \neq 0$ in any direction $\hat \theta$ which points outward from the boundary.}
    Fortunately, there is a remedy which allows the backprojection-convolution formula to be salvaged. 
    The singularities are eliminated by exploiting the theoretical symmetry $m(\hat \theta, r_\perp) = m(-\hat \theta, r_\perp)$, i.e. the assumption that measurements are the same in either direction along a projection line. Theoretically this is a valid assumption, and experimentally it is expected to hold at least approximately.
    
    To mitigate the incidence of these problematic singularities, we start again with the original expression for the backprojection,
    \begin{equation}
                \Biggl\llbracket_{\mathbb R^n}^{\mathbb C} r \mapsto \oint_{\mathrlap{S^{n-1}}} \; \mathrm d^{n-1} \hat \theta \; d(\hat \theta) m(\hat \theta, r_\perp) \Biggr\rrbracket \, ,
    \end{equation}
    then we symmetrise it so that it integrates over \emph{undirected lines} pointing in the direction of $\hat \theta$:
    \begin{equation}
        \Biggl\llbracket_{\mathbb R^n}^{\mathbb C} r \mapsto \frac 1 2 \oint_{\mathrlap{S^{n-1}}} \; \mathrm d^{n-1} \hat \theta \; \left(d(\hat \theta) m(\hat \theta, r_\perp) + d(-\hat \theta) m(-\hat \theta, r_\perp)\right) \Biggr\rrbracket \, .
    \end{equation}
    Now, for each $\hat \theta$, measurements are contributed to the integral from both sides of $r$, in the directions of $\pm \hat \theta$.
    A volume element $d^n x$ contributes to the integral over $\mathrm d \hat \theta$ in both the direction of $(r-x)$ and in the opposite direction $(x-r)$.
    The volume differential is equivalent to either of two differentials at different values of $\hat \theta$ viz.
    \begin{equation}
        {\mathrm d^n x \, \chi(x)}
        \equiv 
        \mathrm d^{n-1} \hat \theta_+ 
            \int_{0}^\infty \mathrm d s \; s^{n-1} \chi(r + s \hat \theta_+) 
        \equiv \mathrm d^{n-1} \hat \theta_- 
            \int_{0}^\infty \mathrm d s \; s^{n-1} \chi(r + s \hat \theta_-)\, .
    \end{equation}
    There is the equality $d^{n-1} \hat \theta_- = d^{n-1} \hat \theta_+$ since $\hat \theta_- = - \hat \theta_+$, so by taking the average of the two differentials we may write
    \begin{subequations}
    \begin{align}
        {\mathrm d^n x \, \chi(x)}
        &\equiv \frac 1 2
        \mathrm d^{n-1} \hat \theta 
            \left(
            \int_{0}^\infty \mathrm d s \; s^{n-1} \chi(r + s \hat \theta)
            + \int_{0}^\infty \mathrm d s \; s^{n-1} \chi(r - s \hat \theta)
            \right) \\
        &\equiv 
        \frac 1 2 \mathrm d^{n-1} \hat \theta 
            \int_{-\infty}^\infty \mathrm d s \; |s|^{n-1} \chi(r + s \hat \theta) \, ,
    \end{align}
    \end{subequations}
    and the weighted backprojection becomes
    \begin{equation}
                \Biggl\llbracket_{\mathbb R^n}^{\mathbb C} \!\!\!\!\! r \!\mapsto\!\! \int_{\mathrlap{\mathbb R^n}} \; \mathrm d^n x \, \chi(x) \;
                \left( \frac 1 2\int_{-\infty}^\infty \mathrm d s \; |s|^{n-1} \chi(r + s \hat \theta) \right)^{-1} \;
                \!\!\! \frac 1 2 \left(\!d(\hat \theta) m(\hat \theta, x) + d(-\hat \theta) m(-\hat \theta, x)  \!\right) \! \Biggr\rrbracket
                 .
    \end{equation}
    We then use the symmetry $m(\hat \theta, x) = m(-\hat \theta, x)$ to arrive at
    \begin{equation}
                \Biggl\llbracket_{\mathbb R^n}^{\mathbb C} r \mapsto \int_{\mathrlap{\mathbb R^n}} \; \mathrm d^n x \, \chi(x) \;
                \left( \frac 1 2\int_{-\infty}^\infty \mathrm d s \; |s|^{n-1} \chi(r + s \hat \theta) \right)^{-1} \;
                \frac 1 2 \left(d(\hat \theta) + d(-\hat \theta) \right) m(\hat \theta, x) \Biggr\rrbracket
                \, .
    \end{equation}     
    Note that if $\chi$ has measure zero along a given line, then there's no opportunity for that line to be backprojected with an integration over $x \in X$. Therefore, this expression for the backprojection can only hold when 
    \begin{equation}
        \forall r,\theta, \left( \int_{-\infty}^\infty \mathrm d s \; \chi(r + s \hat \theta) = 0 \Rightarrow\frac 1 2 \left(d(\hat \theta) + d(-\hat \theta) \right) m(\hat \theta, x) = 0 \right) \, .
    \end{equation}
    This is equivalent to the statement that for the line in direction $\hat \theta$ containing $r$, we must have either $m(\hat \theta, r) = 0 $ or $\left(d(\hat \theta) + d(-\hat \theta) \right) = 0$ or $\int_{-\infty}^\infty \mathrm d s \; \chi(r + s \hat \theta) \neq 0$. We assume that lines not intersecting $V$ have $m = 0$.
\end{proof}

\section{Conclusion}
This article provides a self-contained proof for inversion formulae for certain kinds of partial X-ray transforms, in $n \geq 2$ dimensions. While these inversion formulae constitute pure mathematical results which may be of interest to some, our primary motivation was to produce inversion formulae which are computationally expedient for use in X-ray computerised tomography. This has been accomplished by phrasing the inversion formulae as convolution-backprojection and backprojection-convolution methods. Examples have been provided on how these formulae specialise to particular parallel-beam and cone-beam scanning geometries. A symbolic computer implementation for the cone-beam example in \S\ref{sec:ex_cone_cylinder} (cylindrical source locus) is described in \cite{grewar2024preprint}.

\section*{Acknowledgments}\label{sec:acknowledgments}
This research was funded by an Australian Government Research Training Program (RTP) Scholarship.
Thanks to my supervisors Glenn R. Myers and Andrew M. Kingston for their general support and supervision.
Thanks to Zbigniew Stachurski for encouraging me to increase the mathematical rigour of this work and to formalise the proofs.
Thanks to Lindon Roberts for guidance on how to structure a mathematical article of this kind.



\bibliographystyle{unsrt} 
\bibliography{refs}

\vfill

\end{document}